\documentclass[12pt]{article}
\usepackage{graphicx}
\usepackage{amsmath}
\usepackage{graphicx}
\usepackage{amsfonts}
\usepackage{amssymb}
\usepackage{amsmath, amssymb ,amsthm, amsfonts, amsgen}
\DeclareGraphicsExtensions{.eps,.bmp,.jpg,.pdf,.mps,.png,.gif}
\numberwithin{equation}{section} 
  \newtheorem{Theorem}{Theorem}[section]
\newtheorem{Lemma}[Theorem]{Lemma}
\newtheorem{Proposition}[Theorem]{Proposition}

\newcommand\AMSname{AMS subject classifications}
\begin{document}

\title{Ferromagnetic thin multi-structures}
\author{Antonio Gaudiello\footnote{DIEI,
Universit\`a degli Studi di Cassino e del Lazio Meridionale, via
G. Di Biasio 43, 03043 Cassino (FR), Italia. e-mail:
gaudiell@unina.it} $\,$and Rejeb Hadiji\footnote{Universit\'e Paris-Est, LAMA,
 Laboratoire d'Analyse et de Math\'ematiques Appliqu\'ees, UMR 8050, UPEC, F-94010, Cr\'eteil, France. e-mail: hadiji@u-pec.fr}}\date{ }
\maketitle

\begin{abstract}
In this paper, starting from  the classical $3D$ non-convex and
nonlocal micromagnetic energy for ferromagnetic materials, we
determine, via an asymptotic analysis, the free energy of a
multi-structure consisting of a nano-wire in junction with a  thin
film and of a  multi-structure consisting of two joined
nano-wires. We assume that the volumes of the two parts composing
each multi-structure vanish with  same rate. In the first case, we
obtain a $1D$ limit problem on the nano-wire and a $2D$ limit problem
on the thin film, and the two limit problems are uncoupled.
 In the second case, we obtain two $1D$
limit problems coupled by a junction condition on the
magnetization. In both cases, the limit problem remains
non-convex, but now it becomes
 completely local.

\medskip

\medskip

\par
\noindent2000 \AMSname: 78A25, 49S05, 78M35
\end{abstract}

\section{Introduction}
In this paper, starting from  the classical $3D$ micromagnetic
energy for ferromagnetic materials (see L. D. Landau and E. M.
Lifshitz \cite{LL} and  W. F. Brown \cite{Br}), we determine, via
an asymptotic analysis, the free energy of a  multi-structure
consisting of a nano-wire in junction with a  thin film and of a
multi-structure consisting of two joined nano-wires. These
multi-structures appear in nano electronic devices (for instance,
see \cite{FA} and \cite{OK}). For reasons of simplicity and
economy, especially by a numerical point of view, one tries to
reshape three-dimensional multi-structures, with multi-structures
having a smaller size in thin components. 
\smallskip

In the sequel, $x=(x_1, x_{2},x_3)$ denotes the generic point of
$\mathbb R^3$. If $\eta_1$, $\eta_2$, $\eta_3\in \mathbb{R}^3$,
then $ (\eta_1\vert \eta_2\vert \eta_3)$ denotes the $3 \times 3$
real matrix having $\eta_1^T$ as first column, $\eta_2^T$ as
second column, and $\eta_3^T$ as third column. In according to
this notation, if $v:A\subset
\mathbb{R}^3\rightarrow\mathbb{R}^3$, then $D v$ denotes the $3
\times 3$ real matrix $\left(D_{x_1} v| D_{x_2} v\vert D_{x_3}
v\right)$, where $D_{x_i}v\in \mathbb{R}^3$, i=1,2,3, stands for
the derivative of $v$ with respect to $x_i$.\smallskip

Let  $ \left\{h_n\right\}_{n \in \mathbb N}\subset]0,1[$ be a
vanishing sequence of positive numbers, and let $\Theta
\subset\mathbb ]0,1[^2$ be an open connected  set with smooth
boundary. In this paper, we consider two kinds of  thin
multistructures in $\mathbb{R}^3$. In the first case, for every
$n\in \mathbb{N}$, we set
$$\Omega_n=\left(h_n\Theta\times [0,1[\right)\cup
 \left(\Theta\times ]-h_n^2,0[\right),$$
 which approximates a wire in junction with a thin film (see Fig. 1), as $n$ diverges. In the
 second case, we set
 $$\Omega_n=\left(]-h_n,0[^2\times [0,1[\right)\times \left(]-h_n,1[\times ]-h_n,0[^2\right)$$
which approximates two joined wires (see Fig. 2), as $n$ diverges.
In both cases, the volumes of the two parts of the multi-structure
vanish with  same rate.
\begin{figure}[h]
\centering
\includegraphics[width=9cm]{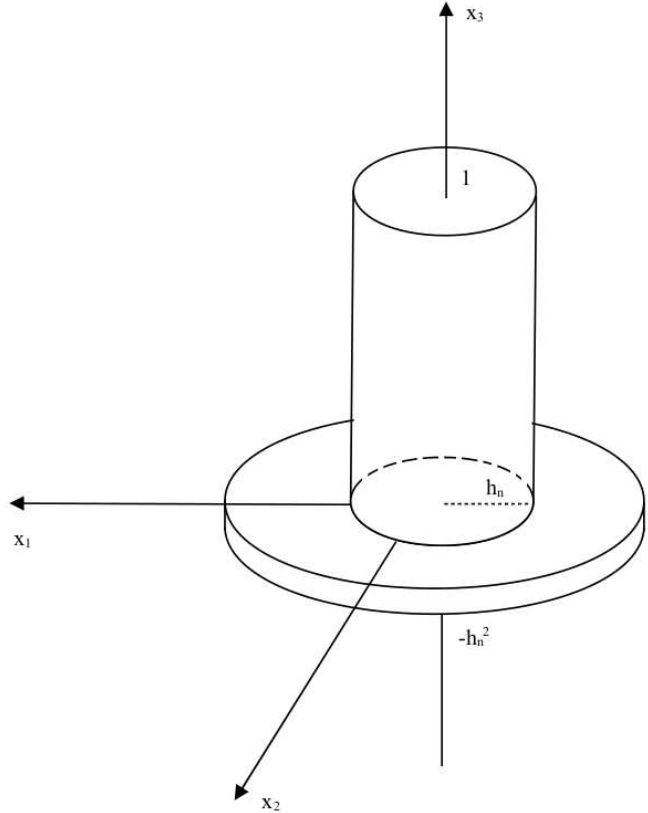}
 \caption{$\Omega_n$ in the case wire - thin film}\label{Fig.1}
\end{figure}
\begin{figure}[h]
\centering
\includegraphics[width=9cm]{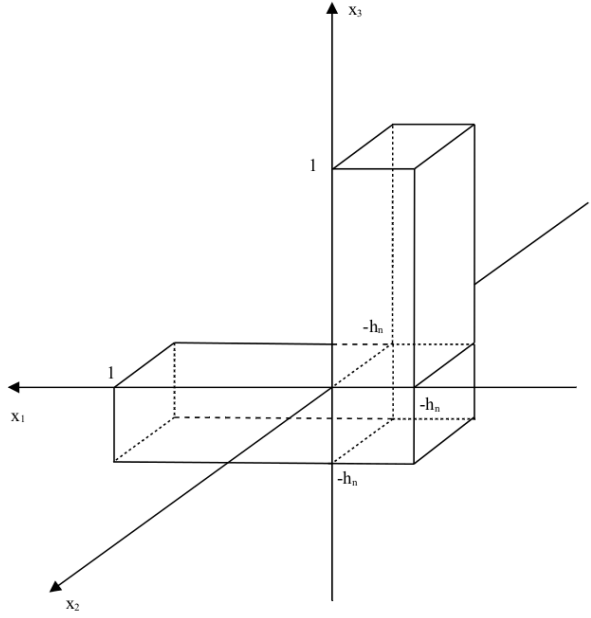}
 \caption{$\Omega_n$ in the case wire - wire}\label{Fig.2}
\end{figure}
 The aim of this paper is to study the asymptotic
behavior, as $n$ diverges, of the following non-convex, nonlocal
variational problem:
\begin{equation}\label{proiniz}\left\{\begin{array}{ll}\begin{array}{ll}
 J_n=\min\Bigg\{\displaystyle{
\int_{\Omega_n}\left(\lambda\vert
DM\vert^2+\varphi(M)+\frac{1}{2}DU_{M} M-2F_nM\right)dx:}\,\,\\
\quad\quad\quad\quad\quad\quad M\in
H^1(\Omega_n,S^2)\Bigg\},\end{array}\\\\\hbox{div}(-DU_{M}+M)=0
\hbox{
 in }\mathbb{R}^3,
\end{array}\right.\end{equation} where  $\lambda$ is a
positive constant, $\varphi: S^2\rightarrow[0,+\infty[$ is a
continuous and even function, $S^2$ denotes the unit sphere of
$\mathbb{R}^3$, and  $F_n\in L^2(\Omega_n, \mathbb R^3)$. It is
understood that $M=0$ in
$\mathbb{R}^3\setminus\Omega_n$.\smallskip

 In classical theory of micromagnetics,
$M:\Omega_n\rightarrow\mathbb{R}^3$ denotes the magnetization and
the body is always locally magnetized to a saturation
magnetization $\vert M(x)\vert=m(T)>0$  unless the local
temperature $T$ is greater or equal to Curie temperature depending
on the body. In the latter case $m(T)=0$, and the material ceases
to behave ferromagnetically. In this paper, we suppose constant
temperature lower than Curie temperature and, without loss of
generality, we assume that $m=1$, that is $M(x)\in S^2$. The
exchange energy $\int_{\Omega_n}\vert DM\vert^2dx$ penalizes the
spatial variation of $M$, driving the body to have large regions
of uniform magnetization separated by thin transition layers. The
scalar function $U_{M}:\mathbb{R}^3\rightarrow\mathbb{R}$ is the
so-called magnetostatic potential. The magnetostatic energy
$\int_{\Omega_n}DU_{M} Mdx=\int_{R^3}\vert DU_{M}\vert^2dx$ favors
div$M=0$  in $\Omega_n$ and $M\cdot \nu=0$ on $\partial \Omega_n$,
where $\nu$ is the exterior unit normal to $\partial \Omega_n$.
The constant $\lambda$ is typically on order of 100 nanometers and
measures the relative strength of exchange energy with respect to
the magnetostatic energy. The anisotropy energy
$\int_{\Omega_n}\varphi(M)dx$ favors   magnetization along special
crystallographic directions, while the external (Zeeman) energy
$\int_{\Omega_n}F_nMdx$ favors magnetization parallel to an
externally applied field.\smallskip

Reformulating the problem on a fixed domain through appropriate
rescalings of the kind proposed by P. G. Ciarlet and P. Destuynder
\cite{CD}, imposing appropriate convergence assumptions on the
rescaled exterior fields and  using   the main ideas of
$\Gamma$-convergence method introduced by E. De Giorgi \cite{DGF},
we derive the limit problem  in both previous cases. Specifically,
in the  case: wire - thin film, we prove that (see Theorem
\ref{ultimo})
\begin{equation}\nonumber\begin{array}{l}\displaystyle{\lim_n\frac{J_n}{h^2_n}=\min\Bigg\{\vert\Theta\vert\int_0^1\left(\lambda\left\vert
\frac{d\mu^a}{dx_3}\right\vert^2+\varphi(\mu^a)-\frac{2}{\vert\Theta\vert}F^a\mu^a
 \right)dx_3+}\\\\\displaystyle{
\frac{1}{2}\Bigg(\alpha(\Theta)\int_0^1\vert\mu^a_1\vert^2dx_3+\beta(\Theta)\int_0^1\vert\mu^a_2\vert^2dx_3+ \gamma(\Theta)\int_0^1\mu^a_1\mu^a_2dx_3\Bigg):}\\\\
\quad\quad \mu^a=(\mu^a_1,\mu^a_2,\mu^a_3)\in
H^1\left(\left]0,1\right[,S^2\right)\Bigg\}+\\\\
\displaystyle{\min\Bigg\{\int_{\Theta}\left(\lambda\left\vert
D\mu^b\right\vert^2+ \varphi(\mu^b)+ \frac{1}{2}\vert
\mu^b_3\vert^2-2F^b\mu^b\right)dx_1dx_2:}\\\\
\quad\quad \mu^b=(\mu^b_1,\mu^b_2,\mu^b_3)\in
H^1\left(\Theta,S^2\right)\Bigg\}.
\end{array}\end{equation}
In the  case: wire - wire, we prove that
 (see Theorem \ref{ultimow-w})
\begin{equation}\nonumber\begin{array}{l}\displaystyle{\lim_n\frac{J_n}{h_n^2}=\min\Bigg\{\int_0^1\left(\lambda\left\vert
\frac{d\mu^a}{dx_3}\right\vert^2+\varphi(\mu^a)-2F^a\mu^a
 \right)dx_3+}\\\\\displaystyle{
\frac{1}{2}\Bigg(\alpha(]-1,0[^2)\int_0^1\vert\mu^a_1\vert^2dx_3+\beta(]-1,0[^2)\int_0^1\vert\mu^a_2\vert^2dx_3+ \gamma(]-1,0[^2)\int_0^1\mu^a_1\mu^a_2dx_3\Bigg)+}\\\\
\displaystyle{\int_0^1\left(\lambda\left\vert
\frac{d\mu^b}{dx_1}\right\vert^2+\varphi(\mu^b)-2G^b\mu^b
 \right)dx_1+}\\\\\displaystyle{
\frac{1}{2}\Bigg(\alpha(]-1,0[^2)\int_0^1\vert\mu^b_2\vert^2dx_1+\beta(]-1,0[^2)\int_0^1\vert\mu^b_3\vert^2dx_1+ \gamma(]-1,0[^2)\int_0^1\mu^b_2\mu^b_3dx_1\Bigg):}\\\\
(\mu^a,\mu^b)=\left((\mu^a_1,\mu^a_2,\mu^a_3),(\mu^b_1,\mu^b_2,\mu^b_3)\right)\in
H^1\left(\left]0,1\right[,S^2\right)\times
H^1\left(\left]0,1\right[,S^2\right), \,\,
\mu^a(0)=\mu^b(0)\Bigg\}.
\end{array}\end{equation}
Above, $F^a(x_3)$ is the integral in $dx_1dx_2$  of the $L^2$-weak
limit of the rescaled external field in the vertical domain,
$F^b(x_1,x_2)$ is the integral in $dx_3$  of the $L^2$-weak limit
of the rescaled external field in the horizontal domain,
$G^b(x_1)$ is the integral in $dx_2dx_3$ of the $L^2$-weak limit
of the rescaled external field in the horizontal domain. To define
coefficients $\alpha$, $\beta$, $\gamma$, if $S\subset\mathbb
R^{2}$ is a bounded open connected set, we introduce the weak
solutions $p$ and $q$, depending on $S$, of the following problems
\begin{equation}\nonumber\left\{\begin{array}{l}
p\in W^1(\mathbb{R}^2),\\\\ \Delta p=0 \hbox{ in }S,\\\\
\Delta p=0 \hbox{ in }\mathbb{R}^2\setminus S,\\\\
\displaystyle{\left[\frac{\partial p}{\partial \nu}\right]=\nu e_1
\hbox{ on }
\partial S,}
\end{array}\right.\quad \left\{\begin{array}{l}
q\in W^1(\mathbb{R}^2),\\\\ \Delta q=0 \hbox{ in }S,\\\\
\Delta q=0 \hbox{ in }\mathbb{R}^2\setminus S,\\\\
\displaystyle{\left[\frac{\partial q}{\partial \nu}\right]=\nu e_2
\hbox{ on }
\partial S,}
\end{array}\right.\end{equation}
where $W^1(\mathbb{R}^2)$ denotes the Beppo-Levi space on
$\mathbb{R}^2$ (see Section \ref{preliminaries}), $\nu$  the
exterior unit normal to $\partial S$,   $\left[\frac{\partial
\cdot}{\partial \nu}\right]$    the jump of $\frac{\partial
\cdot}{\partial \nu}$ on $\partial S$,   and $e_1=(1,0)$,
$e_2=(0,1)$. Then, we set
\begin{equation}\label{abcstarrrr}\begin{array}{l}\displaystyle{ \alpha(S)= \int_{\mathbb{R}^2}\vert
Dp\vert^2dydz,\quad \beta(S)=\int_{\mathbb{R}^2}\vert
Dq\vert^2dydz, \quad \gamma(S)=2\int_{\mathbb{R}^2} DpDqdydz}
,\end{array}\end{equation} where $(y,z)$ denote the coordinates in
$\mathbb{R}^2$. We remark that, if $S$ is  sufficiently smooth,
definitions in (\ref{abcstarrrr}) are equivalent to
\begin{equation}\nonumber\begin{array}{l}\displaystyle{\alpha(S)= \int_{\partial S} p\nu e_1ds,\quad \beta(S)=\int_{\partial
S} q\nu e_2ds, \quad \gamma(S)=\int_{\partial S} q\nu e_1
ds+\int_{\partial S} p\nu e_2ds .}
\end{array}\end{equation}

If $S=\{(x_1,x_2)\in \mathbb{R}^2: x_1^2+x_2^2<1\}$, it results
that $\alpha(S)=\beta(S)=\frac{\pi}{2}$ and $\gamma(S)=0$ (see
Theorem 3.1 in \cite{Sa}).

In the case: wire - thin film, we obtain a $1D$ limit problem on
the wire and a $2D$ limit problem on the thin film, and the two
limit problems are uncoupled. In particular, if
$\Theta=\{(x_1,x_2)\in \mathbb{R}^2: x_1^2+x_2^2<1\}$,
$\varphi=0$, $F^a=0$ and $F^b=0$, then the minimum in the wire is
attained by $(0,0,1)$ or $(0,0,-1)$, while the minimum in the thin
film is attained by every constant $S^2$-vector  parallel to the
thin film.

 In the case: wire - wire, we obtain two $1D$
limit problems coupled by the junction condition on the
magnetization $\mu^a(0)=\mu^b(0)$.

In both cases, the limit problem remains non-convex, but now it
becomes
 completely local.
 Strong convergences in
  $H^1$-norm are obtained for the rescaled
  magnetization.\smallskip

In Section \ref{preliminaries}, we recall the definition and some
properties of the Beppo Levi space on $\mathbb{R}^2$. In Section
\ref{w-f}, we study the case wire - thin film. We use two
different rescalings: one for the wire and a second one for the
thin film. The main difficulty is to identify the limit of the
magnetostatic energy. While it is quite classical in the thin film
where only the component of the magnetization orthogonal to the
film appears in the limit (see \cite{GiJa}), it becomes more
complicated in the wire where  the following combination of  the
first two components of the magnetization with coefficients
involving solutions of PDE in Beppo Levi space on $\mathbb{R}^2$
intervene:
$\alpha(\Theta)\int_0^1\vert\mu^a_1\vert^2dx_3+\beta(\Theta)\int_0^1\vert\mu^a_2\vert^2dx_3+
\gamma(\Theta)\int_0^1\mu^a_1\mu^a_2dx_3$. These coefficients
depend on the geometry of the cross section of the wire. We
explicitly remark that, to our knowledge, we are the first to obtain this explicit formula
for  a wire with a generic cross section. Finally, using the $\Gamma$-convergence method with
suitable test functions and a density result proved in
\cite{GaH1}, we identify the limit problem which results
uncoupled. In Section \ref{w-w}, we study the case wire - wire,
with wires having rectangular cross section. In this case,
the main difficulty is to obtain the junction condition and to perform the limit of the magnetostatic
energy. To this aim we have to use different and more
sophisticated rescaling and symmetry arguments which, in some
sense, take into account the geometry and that the limit problem
will be coupled. 

Our study  can be easily extended to treat
multi-structures as in Figure 3, or cruciform
multi-structures.\smallskip
\begin{figure}[h]
\centering
\includegraphics[width=9cm]{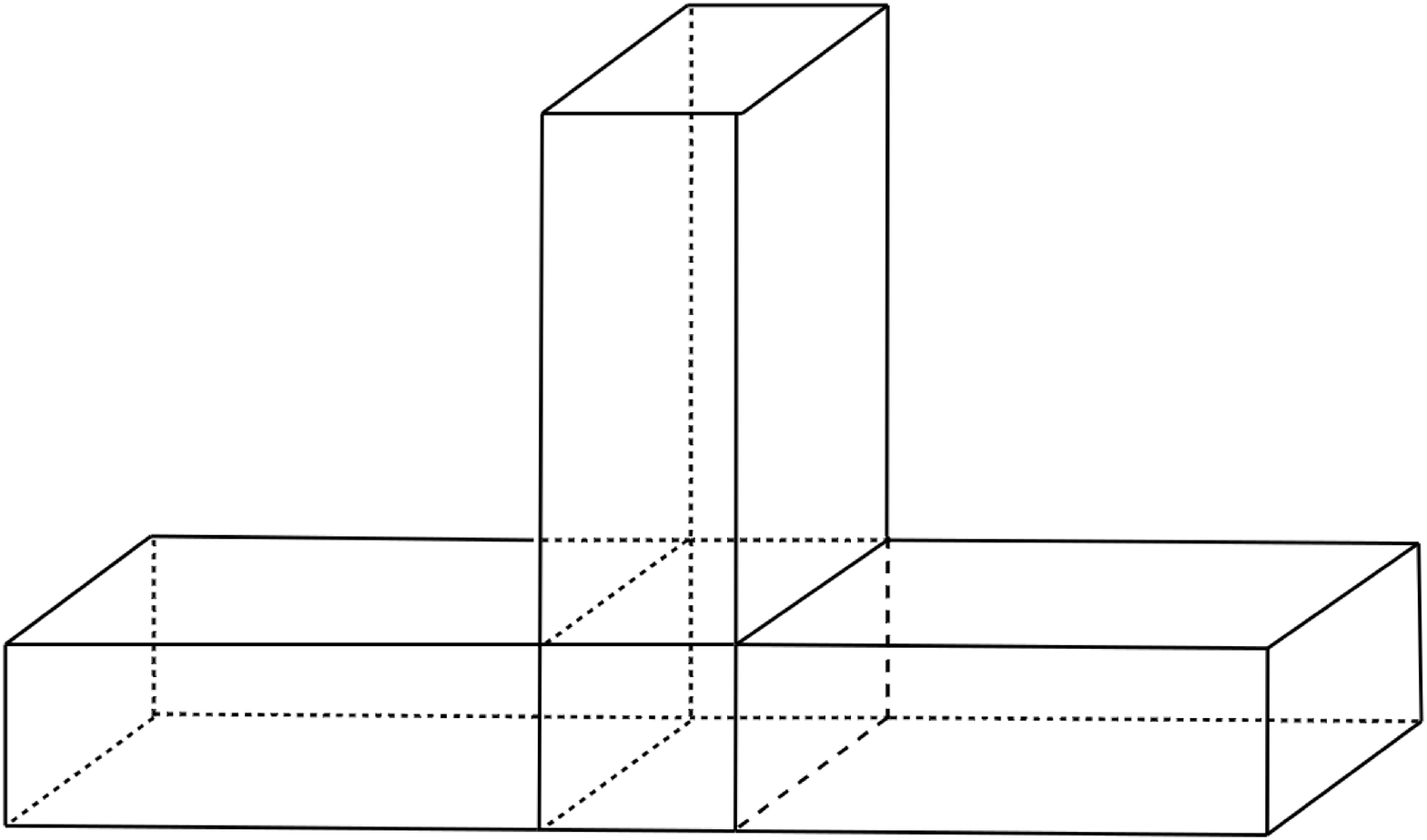}
 \caption{}\label{Fig.3}
\end{figure}

Several results regarding the study of a single ferromagnetic thin
film are present in literature. G. Gioia and R. D. James
\cite{GiJa} were the first to prove that the magnetostatic energy
behaves, at the limit, like an anisotropic local term which forces
the magnetization to be tangent to the thin film. This result was
extended by C. Leone and R. Alicandro \cite{AL} to the case with
non-convex exchange energy, and by M. Ba{\'\i}a and E. Zappale
\cite{BaZa} to a thin film with nonhomogeneous profile.  The case  with degenerative coefficients  was considered by R. Hadiji and K. Shirakawa  \cite{HaSh}. The
time-dependent case was treated by H. Ammari, L. Halpern and K.
Hamdache \cite{AHH}, and by G. Carbou \cite{Car}. F. Alouges, T.
Rivi\`{e}re and  S. Serfaty   \cite{AlRiSe} and C. Rivi\`{e}re and
S. Serfaty \cite{RS} considered an infinite cylinder where the
magnetization does not depend on the vertical coordinate. In
\cite{AlRiSe} the authors showed that bounded-energy
configurations tend to be planar, except in small regions where
one can  observe vortices. In \cite{RS} the magnetization is
moreover constrained to be in the horizontal plane, which avoids
the vortices. F. Alouges and S. Labb\'e \cite{AlLa} proposed a
model of films with strong convergence of minimizers  when the
exchange parameter vanishes and with vertically invariant
configurations on the cylindrical domain. For reproducing the non
uniform states observed experimentally in thin films, very
different regimes were considered by A. Desimone, R.V. Kohn, S.
Muller and F. Otto \cite{DKMO}, and by R.V. Kohn and V.V.
Slastikov in \cite{KS}, where $\frac{h}{l}$ and
$\frac{\lambda}{l}$ vanish, $h$ being the film thickness, $l$ the
in-plane diameter and $\lambda$ the exchange length of the
ferromagnetic material.  Ferroelectric thin films were studied by  A. Gaudiello and K. Hamdache  in \cite{GaHam}.\smallskip

Single ferromagnetic nano-wire with circular cross section and
finite length was studied by G. Carbou and S. Labb\'e \cite{CL}.
In this paper, they also consider a stabilization problem. A
similar model of wire with infinite length was studied by G.
Carbou, S. Labb\'e and E. Tr\'elat \cite{CLT}. Curved  nano-wire was examined by V.V. Slastikov and C. Sonnenberg in \cite{SS}.\smallskip

 In
\cite{GaHa3} we considered two joined ferromagnetic thin films and
we proved that the limit magnetizations are coupled when the
volumes of the two thin films vanish with the same rate.

Multi-structures like in this paper were considered in \cite{GaH1}
and \cite{GaH2}, where we developed an asymptotic analysis of
minimizing maps with values in $S^2$ for the energy
$\int_{\Omega_n}(\vert DM\vert^2-2F_nM)dx$, neglecting the term
with the nonlocal magnetostatic energy which characterizes the
actual paper.

\section{Preliminaries\label{preliminaries}}
Let \begin{equation}\nonumber W^1(\mathbb{R}^2)=\left\{ \phi\in
L^2_{\hbox{loc}}(\mathbb{R}^2): D\phi\in
\left(L^2(\mathbb{R}^2)\right)^2\right\}/\mathbb{R}\end{equation}
equipped with the inner product
\begin{equation}\label{innprod}(\phi_1,\phi_2)\in W^1(\mathbb{R}^2)\times
W^1(\mathbb{R}^2)\rightarrow\int_{\mathbb{R}^2}
D\phi_1D\phi_2dydz,\end{equation} where $(y,z)$ denote the
coordinates in $\mathbb{R}^2$. It is well known that
$W^1(\mathbb{R}^2)$ is a Hilbert space (see \cite{DeLi}, Corol.
1.1) and it is separable. Consequently, if $S\subset\mathbb R^{2}$
is a bounded open set, every one of the following problems
\begin{equation}\label{resequxx1}\left\{\begin{array}{l}p\in W^1(\mathbb{R}^2),\\\\ \displaystyle{\int_{\mathbb{R}^2} Dp
D\phi\, dydz=\int_{S}D_{y}\phi \,dydz,\quad\forall \phi\in
W^1(\mathbb{R}^2)},\end{array}\right.\end{equation}
\begin{equation}\label{resequxx2}\left\{\begin{array}{l}q\in W^1(\mathbb{R}^2),\\\\ \displaystyle{\int_{\mathbb{R}^2}
Dq
D\phi\, dydz=\int_{S}D_{z}\phi\, dydz,\quad\forall \phi\in
W^1(\mathbb{R}^2)},\end{array}\right.\end{equation}
\begin{equation}\label{resequxx}\left\{\begin{array}{l}p_c\in W^1(\mathbb{R}^2),\\\\ \displaystyle{\int_{\mathbb{R}^2} Dp_c
D\phi\, dydz=\int_{S}cD\phi \,dydz,\quad\forall \phi\in
W^1(\mathbb{R}^2),}\end{array}\right.\end{equation} with
$c=(c_1,c_2)\in \mathbb{R}^2$,  admits a unique solution which
obviously depends on $S$. Then, we set
\begin{equation}\label{abc}\begin{array}{l}\displaystyle{ \alpha(S)= \int_{\mathbb{R}^2}\vert
Dp\vert^2dydz,\quad \beta(S)=\int_{\mathbb{R}^2}\vert
Dq\vert^2dydz, \quad \gamma(S)=2\int_{\mathbb{R}^2} DpDqdydz}
.\end{array}\end{equation}

In the sequel, we shall use the following  evident result.
\begin{Lemma}\label{pc} Let $p$ and $q$ be the unique solutions of
(\ref{resequxx1}) and (\ref{resequxx2}), respectively. Then, for
every $c=(c_1,c_2)\in \mathbb{R}^2$, the unique solution $p_c$ of
(\ref{resequxx}) is given by:
$$p_c=c_1p+c_2q. $$
\end{Lemma}

We  recall the  Poincar\'e Lemma (which is well known if the
domain is bounded).
\begin{Lemma} Let $\xi\in \left(L^2(\mathbb{R}^2)\right)^2$ such that rot $\xi=0$.
Then, there exists a unique $w\in W^1(\mathbb{R}^2)$ such that
$\xi=Dw$.
\end{Lemma}
\begin{proof} The fact that rot $\xi=0$ provides the existence of
$T\in {\cal D}'(\mathbb{R}^2)$  such that $\xi=DT$, and $T$ is
unique up to a constant (see \cite{Sc}, Ch. II, Th. VI, page 59).
On the other hand, since $\xi\in
\left(L^2(\mathbb{R}^2)\right)^2$, Kryloff Theorem assures that
$T\in L^2_{loc}(\mathbb{R}^2)$ (see \cite{Sc}, Ch. VI, Th. XV,
page 181).\end{proof}

The following result was suggested by F. Murat \cite{M}.
\begin{Proposition}\label{Murat} Let $u\in L^2_{loc}(\mathbb{R}^2)$ be such that
$Du\in \left(L^2(\mathbb{R}^2)\right)^2$. Then, there exist a
sequence $\{\varphi_n\}_{n\in \mathbb{N}}\subset
C_0^\infty(\mathbb{R}^2)$ such that $D\varphi_n\rightarrow Du$
strongly in $\left(L^2(\mathbb{R}^2)\right)^2$.
\end{Proposition}

For sake of completeness, we conclude this section giving another
representation of $W^1(\mathbb{R}^2)$. There exists a constant
$c>0$, and for every $\phi\in W^1(\mathbb{R}^2)$ there exists
$\overline{\phi} \in \phi$ (we recall that $\phi$ denotes a class
of equivalence) such that (see \cite{Ku}, Th. 6.3)
$$\int_{\mathbb{R}^2}\frac{
\overline{\phi}^2}{\left(1+\log\sqrt{\vert
x\vert^2+1}\right)^2(\vert x\vert^2+1)}dxdy\leq c
\int_{\mathbb{R}^2}\vert D \phi\vert^2dxdy.$$ Consequently, it
results that
\begin{equation}\nonumber W^1(\mathbb{R}^2)=\left\{ \phi\in
L^2_{\hbox{loc}}(\mathbb{R}^2): \frac{
\phi}{\left(1+\log\sqrt{\vert x\vert^2+1}\right)\sqrt{\vert
x\vert^2+1}}\in L^2(\mathbb{R}^2), \quad D\phi\in
\left(L^2(\mathbb{R}^2)\right)^2\right\}/\mathbb{R}\end{equation}
equipped with the inner product in (\ref{innprod}). About this
question see also \cite{Le}.

\section{Wire - thin film\label{w-f}}

This section is devoted to study the asymptotic behavior, as $n$
diverges, of problem (\ref{proiniz}) in the first case, that is
the case  wire - thin film.

\subsection{The setting of the problem  }
 Let $\Theta
\subset\mathbb ]0,1[^2$ be an open connected  set with smooth
boundary and,  for every  $n \in \mathbb N$, let
$\Omega_n^a=h_n\Theta\times [0,1[$,  $\Omega_n^b=\Theta\times
]-h^2_n,0[$  and
 $\Omega_n=\Omega_n^a\cup\Omega_n^b$ (see Fig. 1).

 Let $B=]-2, 2[^3$,  and  set
\begin{equation}\label{calU} {\cal U}=\left\{U\in L^1_{loc}(\mathbb{R}^3)\,:\,U\in L^2(B), \,\, DU\in
(L^2(\mathbb{R}^3))^3,\,\,\int_{B}Udx=0 \right\}.\end{equation} It
is easy to prove that ${\cal U}$ is contained in $
L^2_{loc}(\mathbb{R}^3)$ and it is a Hilbert space with the inner
product: $(U,V)= \int_{\mathbb{R}^3}DUDVdx+\int_{B}UVdx$.
Moreover, it follows from the Poincar\'e-Wirtinger inequality that
a norm on ${\cal U}$ equivalent to $(U,U)^{\frac{1}{2}}$ is given
by $\left(\int_{\mathbb{R}^3}\vert
DU\vert^2dx\right)^{\frac{1}{2}}$. Then, Lax-Milgram theorem
provides that, for  $M\in L^2(\Omega_n, \mathbb{R}^3)$, the
following equation
\begin{equation}\label{resequ}\left\{\begin{array}{l}U_{M,n}\in {\cal U},\\\\ \displaystyle{\int_{\mathbb{R}^3} DU_{M,n}DU
dx=\int_{\Omega_n}{M}DU dx,\quad\forall U\in {\cal
U,}}\end{array}\right.\end{equation}
 admits a unique solution  and $U_{M,n}$ is characterized as
 the unique minimizer of the following problem
\begin{equation}\label{resfun}\min\left\{\frac{1}{2}\int_{\mathbb{R}^3}\vert
DU-{M}\vert^2dx:U\in {\cal U}\right\},\end{equation} where it is
understood that $M=0$ in $\mathbb{R}^3\setminus\Omega_n$.
Moreover, $U_{M}$ belongs to $H^1(\mathbb{R}^3)$ up to an additive
constant (see \cite{JaKi}).

 Let $\lambda$ be a
positive constant, $\varphi: S^2\rightarrow[0,+\infty[$ be a
continuous, even function and, for every $n\in \mathbb{N}$,
$F_n\in L^2(\Omega_n, \mathbb R^3)$. The following problem:
\begin{equation}\label{minoriginale}\min\left\{
\int_{\Omega_n}\left(\lambda\vert
DM\vert^2+\varphi(M)+\frac{1}{2}DU_{M,n} M-2F_nM\right)dx:M\in
H^1(\Omega_n,S^2)\right\}\end{equation} has at least one solution
(see \cite{V}). In general, one can not expect a unique solution,
because of the non-convexity of the constraint $M(x)\in S^2$. The
aim of this section is  to study the asymptotic behavior, as $n$
diverges, of problem (\ref{minoriginale}).

\subsection{The rescaled problem }

By setting
\begin{equation}\nonumber\left\{\begin{array}{ll}
\mathbb{R}^3_a=\{(x_1,x_2,x_3)\in \mathbb{R}^3: x_3>0\},\\\\
\mathbb{R}^3_{b}=\{(x_1,x_2,x_3)\in \mathbb{R}^3:
 x_3< 0\},\end{array}\right.\end{equation}
For every $n\in \mathbb{N}$, problem (\ref{minoriginale}) will be
reformulated
 on a fixed domain through the following rescaling:
\begin{equation}\nonumber(x_1,x_2,x_3)\in \mathbb{R}^3\rightarrow \left\{\begin{array}{ll}(h_nx_1,h_nx_2,x_3), \hbox{ if }(x_1,x_2,x_3)\in
\mathbb{R}^3_a,\\\\(x_1,x_2,h_n^2x_3), \hbox{ if }
(x_1,x_2,x_3)\in \mathbb{R}^3_{b}.
\end{array}\right.\end{equation}
Namely, setting
$$\Omega^a=\Theta \times ]0,1[,\quad
\Omega^b=\Theta\times ]-1,0[,$$ and
$$ B^a_n=\left]-\frac{2 }{h_n}, \frac{2
}{h_n}\right[^2\times ]0,2[,\quad B^b_n= ]-2, 2[^2\times
\left]-\frac{2}{h^2_n},0\right[,\quad\forall n\in \mathbb{N},$$
the space ${\cal U}$ defined in (\ref{calU}) is rescaled in the
following
\begin{equation}\label{spazioriscalatoU}\begin{array}{ll} {\cal U}_n= \big\{&u=(u^a, u^b)
\in L^1_{loc} (\overline{\mathbb{R}^3_a})\times L^1_{loc}
(\overline{\mathbb{R}^3_b})\,\,:\,\, (u^a_{|_{B^a_n}},
u^b_{|_{B^b_n}}) \in L^2 (B^a_n)\times L^2 (B^b_n),\\\\&
(Du^a,Du^b)\in
(L^2(\mathbb{R}^3_a))^3\times (L^2(\mathbb{R}^3_b))^3,\quad\displaystyle{\int_{B^a_n}u^adx+\int_{B^b_n}u^bdx=0,}\\\\
& u^a(x_1,x_2,0)= u^b(h_nx_1,h_nx_2,0),\hbox{ for }(x_1,x_2)
\hbox{ a.e. in }\mathbb{R}^2 \big\}.\end{array}
\end{equation} Then,  for every
$\underline{m}=(\underline{m}^a,\underline{m}^b)\in
L^2(\Omega^a,\mathbb{R}^3)\times L^2(\Omega^b,\mathbb{R}^3)$, the
following equation
\begin{equation}\label{equariscalata}\left\{\begin{array}{l}u_{\underline{m},n}=(u_{\underline{m},n}^a,u_{\underline{m},n}^b)\in {\cal U}_n,\\\\
 \displaystyle{\int_{\mathbb{R}^3_a}
\left(\frac{1}{h_n}D_{x_1}u_{\underline{m},n}^a,\frac{1}{h_n}D_{x_2}u_{\underline{m},n}^a,
D_{x_3}u_{\underline{m},n}^a\right)\left(\frac{1}{h_n}D_{x_1}u^a,\frac{1}{h_n}D_{x_2}u^a,
D_{x_3}u^a\right) dx+}\\\\
 \displaystyle{\int_{\mathbb{R}^3_b}
\left(D_{x_1}u_{\underline{m},n}^b,D_{x_2}u_{\underline{m},n}^b,
\frac{1}{h^2_n}D_{x_3}u_{\underline{m},n}^b\right)\left(D_{x_1}u^b,D_{x_2}u^b,
\frac{1}{h^2_n}D_{x_3}u^b\right)
dx=}\\\\\displaystyle{\int_{\Omega^a}\left(\frac{1}{h_n}D_{x_1}u^a,\frac{1}{h_n}D_{x_2}u^a,
D_{x_3}u^a\right){\underline{m}^a} dx+}\\\\
 \displaystyle{\int_{\Omega^b}
\left(D_{x_1}u^b,D_{x_2}u^b,
\frac{1}{h^2_n}D_{x_3}u^b\right){\underline{m}^b} dx,}\quad\forall
u=(u^a,u^b)\in {\cal U}_n,\end{array}\right.\end{equation}
 which  rescales  equation (\ref{resequ}), admits a unique
solution. Its solution
$u_{\underline{m},n}=(u_{\underline{m},n}^a,u_{\underline{m},n}^b)\in
{\cal U}_n $ is characterized as  the unique minimizer of the
following problem rescaling  problem (\ref{resfun}) after the
renormalization by $h^2_n$:
\begin{equation}\label{problemariscalatoj}\begin{array}{l}
j_{\underline{m},n}(u_{\underline{m},n})=\min\left
\{j_{\underline{m},n}(u): u\in {\cal U}_n
\right\},\end{array}\end{equation} where
\begin{equation}\label{funzionaleriscalatoj}\begin{array}{l}\displaystyle{j_{\underline{m},n}:u=(u^a,u^b)\in
{\cal
U}_n\longrightarrow\frac{1}{2}\int_{\mathbb{R}^3_a}\left\vert\left(\frac{1}{h_n}D_{x_1}u^a,\frac{1}{h_n}D_{x_2}u^a,
D_{x_3}u^a\right)-{\underline{m}^a}\right\vert^2
dx+}\\\\
\displaystyle{\frac{1}{2}
\int_{\mathbb{R}^3_b}\left\vert\left(D_{x_1}u^b,D_{x_2}u^b,\frac{1}{h^2_n}
D_{x_3}u^b\right)-{\underline{m}^b}\right\vert^2dx,}
\end{array}\end{equation}
understanding ${\underline{m}^a}=0$ in $\mathbb{R}^3_a\setminus
\Omega^a$ and ${\underline{m}^b}=0$ in $\mathbb{R}^3_b\setminus
\Omega^b$. We note that
$u_{\underline{m},n}=(u_{\underline{m},n}^a,u_{\underline{m},n}^b)$
belongs to $ H^1(\mathbb{R}^3_a)\times H^1(\mathbb{R}^3_b)$ up to
an additive constant.

For every $n\in \mathbb{N}$, $H^1(\Omega_n,S^2)$, $F_n\in
L^2(\Omega_n)$ and the functional involved in problem
(\ref{minoriginale}) renormalized by $h^2_n$ are rescaled in
\begin{equation}\label{spazioriscalatoM}\begin{array}{ll} {\cal M}_n= \Big\{&\underline{m}=(\underline{m}^a, \underline{m}^b)
\in H^1 (\Omega^a,S^2)\times H^1
(\Omega^b,S^2)\,:\\\\
 & \underline{m}^a(x_1,x_2,0)= \underline{m}^b(h_nx_1,h_nx_2,0),\hbox{ for }(x_1,x_2) \hbox{
a.e. in }\Theta \Big\},\end{array}
\end{equation}
\begin{equation}\label{riscalamentoforze}
f_n:x\in \Omega^a\cup \Omega^b\rightarrow f_n(x)=\left\{
\begin{array}{ll} f^{a}_n(x)=
 F_n(h_nx_1,h_nx_2,x_3),\quad \hbox{for }x\hbox{ a.e. in
}\Omega^a, \\\\
 f^{b}_n(x)= F_n(x_1,x_2,h^2_n x_3), \quad
\hbox{for }x\hbox{ a.e. in }\Omega^b,
\end{array}\right.
\end{equation}
and
\begin{equation}\label{funzionaleriscalatoE}\begin{array}{l}\displaystyle{E_n:\underline{m}=(\underline{m}^a,\underline{m}^b)\in
{\cal
M}_n\longrightarrow}\\\\\displaystyle{\int_{\Omega^a}\left(\lambda\left\vert\left(\frac{1}{h_n}D_{x_1}\underline{m}^a|
\frac{1}{h_n}D_{x_2}\underline{m}^a|
D_{x_3}\underline{m}^a\right)\right\vert^2+\varphi(\underline{m}^a)-2f^a_n\underline{m}^a
 \right)dx+}\\\\\displaystyle{\frac{1}{2}\int_{\Omega^a}\left(
\left(\frac{1}{h_n}D_{x_1}u^a_{\underline{m},n},\frac{1}{h_n}D_{x_2}u^a_{\underline{m},n},
D_{x_3}u^a_{\underline{m},n}\right)\underline{m}^a
 \right)dx+}\\\\
\displaystyle{
\int_{\Omega^b}\left(\lambda\left\vert\left(D_{x_1}\underline{m}^b|D_{x_2}\underline{m}^b|\frac{1}{h^2_n}
D_{x_3}\underline{m}^b\right)\right\vert^2+\varphi(\underline{m}^b)-2f^b_n\underline{m}^b\right)dx+}
\\\\
\displaystyle{\frac{1}{2}\int_{\Omega^b}\left(
\left(D_{x_1}u^b_{\underline{m},n},D_{x_2}u^b_{\underline{m},n},
\frac{1}{h^2_n}D_{x_3}u^b_{\underline{m},n}\right)\underline{m}^b\right)dx,}
\end{array}\end{equation}
respectively. Then,  the function defined by
\begin{equation}\nonumber
\begin{array}{ll}
 M_n(h_nx_1,h_nx_2,x_3),\hbox{ for }x\hbox{ a.e. in
}\Omega^a, \quad\quad
 M_n(x_1,x_2,h^2_n x_3),
\hbox{ for }x\hbox{ a.e. in }\Omega^b,
\end{array}
\end{equation}
with $M_n$ solution of problem (\ref{minoriginale}), is a
minimizer of the following problem:
\begin{equation}\label{problemariscalato}\begin{array}{l}
\min\left \{E_n(\underline{m}): \underline{m}\in {\cal M}_n
\right\}.
\end{array}\end{equation}
Actually, the goal becomes to study the asymptotic behavior, as
$n$ diverges, of problem (\ref{problemariscalato}). To this aim,
it will be assumed that
\begin{equation}\label{forze} f^{a}_n\rightharpoonup f^{a}\hbox{ weakly
in } L^2(\Omega^a, \mathbb R^3),\quad\quad f^{b}_n\rightharpoonup
f^{b}\hbox{ weakly in } L^2(\Omega^b,\mathbb R^3).\end{equation}

Note that,  setting for every $n\in \mathbb{N}$
\begin{equation}\label{Emag}\begin{array}{l}\displaystyle{E_n^{mag}:\underline{m}=(\underline{m}^a,\underline{m}^b)\in
L^2(\Omega^a,\mathbb{R}^3)\times
L^2(\Omega^b,\mathbb{R}^3)\longrightarrow}\\\\\displaystyle{\frac{1}{2}\int_{\mathbb{R}^3_a}
\left\vert\left(\frac{1}{h_n}D_{x_1}u^a_{\underline{m},n},\frac{1}{h_n}D_{x_2}u^a_{\underline{m},n},
D_{x_3}u^a_{\underline{m},n}\right)\right\vert^2dx+}\\\\
\displaystyle{\frac{1}{2} \int_{\mathbb{R}^3_b}\left\vert
\left(D_{x_1}u^b_{\underline{m},n},D_{x_2}u^b_{\underline{m},n},
\frac{1}{h^2_n}D_{x_3}u^b_{\underline{m},n}\right)\right\vert^2dx,}
\end{array}\end{equation}
by virtue of   (\ref{equariscalata}), functional $E_n$ can be
rewritten in the following way:
\begin{equation}\label{E+Emag}\begin{array}{l}E_n(\underline{m})=
\displaystyle{\int_{\Omega^a}\left(\lambda\left\vert\left(\frac{1}{h_n}D_{x_1}\underline{m}^a|\frac{1}{h_n}D_{x_2}\underline{m}^a|
D_{x_3}\underline{m}^a\right)\right\vert^2+\varphi(\underline{m}^a)-2f^a_n\underline{m}^a\right)dx+}\\\\
\displaystyle{
\int_{\Omega^b}\left(\lambda\left\vert\left(D_{x_1}\underline{m}^b|D_{x_2}\underline{m}^b|\frac{1}{h^2_n}
D_{x_3}\underline{m}^b\right)\right\vert^2+\varphi(\underline{m}^b)-2f^b_n\underline{m}^b\right)dx+}\\\\\displaystyle{E^{mag}_n(\underline{m}),\quad\forall
\underline{m}=(\underline{m}^a,\underline{m}^b)\in {\cal
M}_n,\quad\forall n\in \mathbb{N}}.
\end{array}\end{equation}

\subsection{The main result\label{ms} }
Let
\begin{equation}\label{limitspace}\begin{array}{ll}{\cal M}=\Big\{\mu=(\mu^a,\mu^b)\in H^1(\Omega^a,S^2)\times
H^1(\Omega^b,S^2)\,:\,
\mu^a\hbox{ is independent of } (x_1,x_2),\\\\
\quad\quad\mu^b\hbox{ is independent of } x_3\Big\}\simeq
H^1(]0,1[,S^2)\times H^1\left(\Theta,S^2\right),
\end{array}\end{equation}
\begin{equation}\label{average}\left\{\begin{array}{ll}F^a: x_3\in ]0,1[\longrightarrow{\displaystyle\frac{1}{\vert \Theta\vert}\int_{\Theta}
f^a(x_1,x_2,x_3)dx_1dx_2},\\\\
F^b:(x_1,x_2)\in \Theta\longrightarrow {\displaystyle\int_{-1}^{0}
f^b(x_1,x_2,x_3)dx_3},
\end{array}\right.\end{equation}
 and
\begin{equation}\label{funlimite}\begin{array}{l}\displaystyle{E:\mu=(\mu^a,\mu^b)=((\mu^a_1,\mu^a_2,\mu^a_3),(\mu^b_1,\mu^b_2,\mu^b_3))\in
{\cal M}\longrightarrow}\\\\\displaystyle{\vert
\Theta\vert\int_0^1\left(\lambda\left\vert
\frac{d\mu^a}{dx_3}\right\vert^2+\varphi(\mu^a)-2F^a\mu^a
 \right)dx_3+}\\\\\displaystyle{
\frac{1}{2}\Bigg(\alpha(\Theta)\int_0^1\vert\mu^a_1\vert^2dx_3+\beta(\Theta)\int_0^1\vert\mu^a_2\vert^2dx_3+\gamma(\Theta)\int_0^1\mu^a_1\mu^a_2dx_3\Bigg)+}\\\\
\displaystyle{\int_{\Theta}\left(\lambda\left\vert
D\mu^b\right\vert^2+ \varphi(\mu^b)+ \frac{1}{2}\vert
\mu^b_3\vert^2-2F^b\mu^b\right)dx_1dx_2,}
\end{array}\end{equation}
where  $\alpha(\Theta)$, $\beta(\Theta)$ and $\gamma(\Theta)$ are
defined by (\ref{abc}) with $S=\Theta$.

This section is devoted to prove the following main result:
\begin{Theorem}\label{ultimo}Assume  (\ref{forze}).
For every $n\in \mathbb{N}$, let
$\underline{m}_n=(\underline{m}_n^a,\underline{m}_n^b)$ be a
solution of (\ref{problemariscalato}) and  $u_n=(u_n^a,u_n^b)$ be
the unique solution of (\ref{problemariscalatoj}) corresponding to
$\underline{m}_n$. Moreover, let ${\cal M}$ and $E$ be defined by
(\ref{limitspace}) and (\ref{funlimite}), respectively. Then,
there exist an increasing sequence of positive integer numbers
$\{n_i\}_{i\in\mathbb N}$ and
$\widehat{\mu}=(\widehat{\mu}^a,\widehat{\mu}^b)=((\widehat{\mu}^a_1,\widehat{\mu}^a_2,\widehat{\mu}^a_3),
(\widehat{\mu}^b_1,\widehat{\mu}^b_2,\widehat{\mu}^b_3))\in{\cal
M}$, depending on the selected subsequence, such that
\begin{equation}\label{convergenzastrong}\begin{array}{l}
\underline{m}^{a}_{n_i}\rightarrow \widehat{\mu}^a\hbox{ stongly
in } H^1(\Omega^a,S^2),\quad\quad
\underline{m}^{b}_{n_i}\rightarrow\widehat{\mu}^b\hbox{ strongly
in } H^1(\Omega^b,S^2),\end{array}
\end{equation}
\begin{equation}\label{convergenzaDstrong}\left\{\begin{array}{ll}
\dfrac{1}{h_{n}}D_{x_1}\underline{m}^{a}_{n}\rightarrow0,\quad\dfrac{1}{h_{n}}D_{x_2}\underline{m}^{a}_{n}\rightarrow
0\hbox{ stongly in }
L^2(\Omega^a,\mathbb{R}^3),\\\\
\dfrac{1}{h^2_{n}}D_{x_3}\underline{m}^{b}_{n}\rightarrow 0\hbox{
strongly in } L^2(\Omega^b,\mathbb{R}^3),\end{array}\right.
\end{equation}
\begin{equation}\label{limcon}\left\{\begin{array}{llll}\dfrac{1}{h_{n_i}}D_{x_1}u_{n_i}^a\rightarrow
\widehat{\xi}^a_1,
 \quad  \dfrac{1}{h_{n_i}}D_{x_2}u_{n_i}^a\rightarrow\widehat{\xi}^a_2, \quad D_{x_3}u_n^a\rightarrow0\quad\hbox{ strongly in
 }L^2(\mathbb{R}^3_a),\\\\D_{x_1}u_n^b\rightarrow0,
 \quad D_{x_2}u_n^b\rightarrow0, \quad \dfrac{1}{h^2_{n_i}}D_{x_3}u_{n_i}^b\rightarrow{\widehat{\mu}_3^b}\quad\hbox{ strongly in
 }L^2(\mathbb{R}^3_b),
\end{array}\right.\end{equation}
as $n$ and $i$ diverge, where $\widehat{\mu}$ is a solution of the
following problem:
\begin{equation}\label{problemafinale}\begin{array}{l}
E(\widehat{\mu})=\min\left \{E(\mu): \mu\in {\cal M} \right\},
\end{array}\end{equation} and
\begin{equation}\label{idxiaiiz} (\widehat{\xi}_1^a,\widehat{\xi}_2^a)(x_1,x_2,x_3)=\left\{\begin{array}{ll} (0,0),
\hbox{ a.e. in } \mathbb{R}^2\times ]1,+\infty[, \\\\
\widehat{\mu}^a_1(x_3)Dp(x_1,x_2)+
\widehat{\mu}^a_2(x_3)Dq(x_1,x_2),\hbox{ a.e. in }
\mathbb{R}^2\times ]0,1[,\end{array}\right.\end{equation} with $p$
(resp. $q$) the unique solution of (\ref{resequxx1}) (resp.
(\ref{resequxx2})).
  It is understood that
      ${\widehat{\mu}_3^b}=0$ in $\mathbb{R}^3_b\setminus \Omega^b$. Moreover,
      the convergence of the energies holds true:
\begin{equation}\label{convenergie}\begin{array}{l}\displaystyle{
\lim_nE_n(\underline{m}_n)= E(\widehat{\mu}).}
\end{array}\end{equation}
\end{Theorem}

\subsection{A Convergence result for  the magnetostatic
energy \label{magpot}}
\begin{Proposition}\label{limvarprin}
Let
$\{\underline{m}_n=(\underline{m}_n^a,\underline{m}_n^b)\}_{n\in\mathbb{
N}}\subset L^2(\Omega^a,\mathbb{R}^3)\times
L^2(\Omega^b,\mathbb{R}^3)$ and
$\mu=(\mu^a,\mu^b)=((\mu^a_1,\mu^a_2,\mu^a_3,),(\mu^b_1,\mu^b_2,\mu^b_3))\in
L^2(\Omega^a,\mathbb{R}^3)\times L^2(\Omega^b,\mathbb{R}^3)$ be
such that $\mu^a$ is independent of $(x_1,x_2)$,  $\mu^b$ is
independent of $x_3$ and
\begin{equation}\label{mnconv}\underline{m}_n^a\rightarrow \mu^a \hbox{ strongly in }L^2(\Omega^a,\mathbb{R}^3),\quad\quad\underline{m}_n^b\rightarrow
\mu^b\hbox{ strongly in }
L^2(\Omega^b,\mathbb{R}^3),\end{equation} as $n$ diverges.
Moreover, for every $n\in\mathbb{ N}$, let $u_n=(u_n^a,u_n^b)$ be
the unique solution of (\ref{problemariscalatoj}) corresponding to
$\underline{m}_n$, and let $E_n^{mag}$ be defined by (\ref{Emag}).
Then, it results that
\begin{equation}\label{conlimprin}\left\{\begin{array}{llll}\dfrac{1}{h_n}D_{x_1}u_n^a\rightarrow
\xi^a_1,
 \quad  \dfrac{1}{h_n}D_{x_2}u_n^a\rightarrow\xi^a_2, \quad D_{x_3}u_n^a\rightarrow0\quad\hbox{ strongly in
 }L^2(\mathbb{R}^3_a),\\\\D_{x_1}u_n^b\rightarrow0,
 \quad D_{x_2}u_n^b\rightarrow0, \quad \dfrac{1}{h^2_n}D_{x_3}u_n^b\rightarrow{\mu_3^b}\quad\hbox{ strongly in
 }L^2(\mathbb{R}^3_b),
\end{array}\right.\end{equation}
as $n$ diverges, where it is understood that
      ${\mu}_3^b=0$ in $\mathbb{R}^3_b\setminus \Omega^b$, and
\begin{equation}\label{idxiaiiz} (\xi_1^a,\xi^a_2)(x_1,x_2,x_3)=\left\{\begin{array}{ll} (0,0),
\hbox{ a.e. in } \mathbb{R}^2\times ]1,+\infty[, \\\\
\mu^a_1(x_3)Dp(x_1,x_2)+ \mu^a_2(x_3)Dq(x_1,x_2),\hbox{ a.e. in }
\mathbb{R}^2\times ]0,1[,\end{array}\right.\end{equation} with $p$
(resp. $q$) the unique solution of (\ref{resequxx1}) (resp.
(\ref{resequxx2})). Furthermore, one has that
\begin{equation}\label{conEmag}\begin{array}{l}\displaystyle{\lim_nE_n^{mag}(\underline{m}_n)=\frac{1}{2}\left(\int_{\mathbb{R}^2\times]0,1[}\vert
\mu^a_1Dp+ \mu^a_2Dq\vert^2dx +
\int_{\Theta}\vert\mu^b_3\vert^2dx_3\right)=}\\\\\displaystyle{
\frac{1}{2}\Bigg(\alpha(\Theta)\int_0^1\vert\mu^a_1\vert^2dx_3+\beta(\Theta)\int_0^1\vert\mu^a_2\vert^2dx_3+
\gamma(\Theta)\int_0^1\mu^a_1\mu^a_2dx_3+\int_{\Theta}\vert
\mu_3^b\vert^2dx_1dx_2\Bigg),}
\end{array}\end{equation}
where  $\alpha(\Theta)$, $\beta(\Theta)$ and $\gamma(\Theta)$ are
defined by (\ref{abc}) with $S=\Theta$.
\end{Proposition}
\begin{proof} The proof will be developed in four steps.

By arguing as in the first part of the proof of proposition 5.1 in
\cite{GaHa3}, one can proves that
\begin{equation}\label{conunL60}\begin{array}{l}
Du_n^a\rightharpoonup 0 \hbox{ weakly in
}(L^2(\mathbb{R}^3_a))^3,\quad Du_n^b\rightharpoonup 0 \hbox{
weakly in }(L^2(\mathbb{R}^3_b))^3,\end{array}\end{equation} as
$n$ diverges, and  that there exist  $\xi^a=(\xi^a_1,\xi^a_2)\in
\left(L^2(\mathbb{R}^3_a)\right)^2$ and $\xi^b\in
L^2(\mathbb{R}^3_b)$ such that, on extraction of a suitable
 subsequence (not relabeled),
\begin{equation}\label{condu}\left\{\begin{array}{l}\displaystyle{\frac{1}{h_n}D_{x_1}u_n^a\rightharpoonup \xi^a_1\hbox{ weakly in
}L^2(\mathbb{R}^3_a),\quad
\frac{1}{h_n}D_{x_2}u_n^a\rightharpoonup \xi^a_2\hbox{ weakly in
}L^2(\mathbb{R}^3_a),}\\\\
\displaystyle{\frac{1}{h^2_n}D_{x_3}u_n^b\rightharpoonup
\xi^b\hbox{ weakly in
}L^2(\mathbb{R}^3_b),}\end{array}\right.\end{equation} as $n$
diverges.

The second step is devoted to identify  $\xi^a$. To this aim,
starting from the following evident relation:
$$ D_{x_2}\left( \frac{1}{h_n}D_{x_1}u_n^a\right)=D_{x_1}\left( \frac{1}{h_n}D_{x_2}u_n^a\right)
\hbox{ in }{\cal D}'(\mathbb{R}^3_a),\quad\forall n \in
\mathbb{N},$$  and using the first two limits in (\ref{condu}),
one obtains that
\begin{equation}\label{D'}\displaystyle{\int_{\mathbb{R}^3_a}\xi_1^aD_{x_2}\varphi dx=
\int_{\mathbb{R}^3_a}\xi_2^aD_{x_1}\varphi dx, \quad\forall
\varphi\in H_0^1(\mathbb{R}^3_a).}\end{equation} By taking
$\varphi(x)=\phi(x_1,x_2)\chi(x_3)$ with $\phi\in
H^1(\mathbb{R}^2)$ and $\chi\in C_0^\infty(]0,+\infty[)$ and
recalling that $H^1(\mathbb{R}^2)$ is separable, it follows from
  (\ref{D'}) that
\begin{equation}\nonumber\left\{\begin{array}{l}\hbox{ for
} x_3 \hbox{ a.e. in }]0,+\infty[,\quad
\displaystyle{\int_{\mathbb{R}^2}\xi_1^a(x_1,x_2,x_3)D_{x_2}\phi(x_1,x_2)
dx_1dx_2=}\\\\ \displaystyle{
\int_{\mathbb{R}^2}\xi_2^a(x_1,x_2,x_3)D_{x_1}\phi(x_1,x_2)
dx_1dx_2, \quad\forall \phi\in
H^1(\mathbb{R}^2).}\end{array}\right.\end{equation} Consequently,
by virtue of the Poincar\'e Lemma (see Section
\ref{preliminaries}), it results that
\begin{equation}\label{poincare}\left\{\begin{array}{l} \hbox{ for
}x_3 \hbox{ a.e. in
}]0,+\infty[,\quad\exists ! w(\cdot,\cdot,x_3)\in W^1(\mathbb{R}^2):\\\\
\xi^a_1(\cdot,\cdot,x_3)=D_{x_1}w(\cdot,\cdot,x_3),\quad
\xi^a_2(\cdot,\cdot,x_3)=D_{x_2}w(\cdot,\cdot,x_3), \quad \hbox{
a.e. in } \mathbb{R}^2.
\end{array}\right.\end{equation}

Now, in equation  (\ref{equariscalata})  with
$\underline{m}=\underline{m}_n$ choose  $u^a=\varphi+c_n$ and
$u^b=c_n$, with $\varphi\in C^\infty_0 (\mathbb{R}^3_a)$ and
$c_n=-(\vert B^a_n \vert+ \vert B^b_n
\vert)^{-1}\int_{B^a_n}\varphi dx$ (such that $(u^a,u^b)\in {\cal
U}_n$).  By multiplying this equation by $h_n$, one has
\begin{equation}\label{equariscalatazlk}\left\{\begin{array}{l}
 \displaystyle{\int_{\mathbb{R}^3_a}
\left(\frac{1}{h_n}D_{x_1}u_n^a,\frac{1}{h_n}D_{x_2}u_n^a,
D_{x_3}u_n^a\right)\left(D_{x_1}\varphi,D_{x_2}\varphi, h_n
D_{x_3}\varphi\right)
dx=}\\\\\displaystyle{\int_{\Omega^a}\left(D_{x_1}\varphi,D_{x_2}\varphi,
h_nD_{x_3}\varphi\right){\underline{m}^a_n} dx},\quad\forall
\varphi\in C^\infty_0
(\mathbb{R}^3_a).\end{array}\right.\end{equation} Then, passing to
the limit, as $n$ diverges, in (\ref{equariscalatazlk}),
convergences  (\ref{mnconv}), (\ref{conunL60}) and (\ref{condu})
give that
$$\int_{\mathbb{R}^3_a}(\xi^a_1,\xi^a_2)(D_{x_1}\varphi,
D_{x_2}\varphi)dx=\int_0^{1}\left((\mu^a_1,\mu^a_2)\int_{\Theta}(D_{x_1}\varphi,
D_{x_2}\varphi)dx_1dx_2\right)dx_3,\quad\forall \varphi\in
C^\infty_0 (\mathbb{R}^3_a).$$ Consequently, arguing as above,
taking into account that $W^1(\mathbb{R}^2)$ is separable, and
using Proposition \ref{Murat}  and (\ref{poincare}), it follows
that, for $x_3$ a.e. in $]1,+\infty[$, $w(\cdot,\cdot,x_3)$ solves
the following problem:
\begin{equation}\nonumber\left\{\begin{array}{l} w(\cdot,\cdot,x_3)\in W^1(\mathbb{R}^2),\\\\\displaystyle{
\int_{\mathbb{R}^2}(D_{x_1}w(x_1,x_2,x_3),D_{x_2}w(x_1,x_2,x_3))(D_{x_1}\phi(x_1,x_2),
D_{x_2}\phi(x_1,x_2))dx_1dx_2=0}, \quad\forall \phi\in
W^1(\mathbb{R}^2),\end{array}\right.\end{equation} while,  for
$x_3$ a.e. in $]0,1[$, $w(\cdot,\cdot,x_3)$ solves the following
one:
\begin{equation}\label{17giugno}\left\{\begin{array}{l} w(\cdot,\cdot,x_3)\in W^1(\mathbb{R}^2),\\\\\displaystyle{
\int_{\mathbb{R}^2}(D_{x_1}w(x_1,x_2,x_3),D_{x_2}w(x_1,x_2,x_3))(D_{x_1}\phi(x_1,x_2),
D_{x_2}\phi(x_1,x_2))dx_1dx_2=}\\\\
\displaystyle{(\mu^a_1(x_3),\mu^a_2(x_3))\int_{\Theta}(D_{x_1}\phi(x_1,x_2),
D_{x_2}\phi(x_1,x_2))dx_1dx_2,\quad\forall \phi\in
W^1(\mathbb{R}^2).}\end{array}\right.\end{equation} Then, by
virtue of Lemma \ref{pc}, it results that, for $x_3$ a.e. in
$]0,+\infty[$,
\begin{equation}\label{idp} w(\cdot,\cdot,x_3)=\left\{\begin{array}{ll} 0,
\hbox{ a.e. in } \mathbb{R}^2, \hbox{ if }x_3>1, \\\\
\mu^a_1(x_3)p(\cdot,\cdot)+ \mu^a_2(x_3)q(\cdot,\cdot),\hbox{ a.e.
in } \mathbb{R}^2, \hbox{ if
}x_3<1,\end{array}\right.\end{equation} with $p$ (resp. $q$) the
unique solution of (\ref{resequxx1}) (resp. (\ref{resequxx2})).

Finally, since Tonelli theorem assures that $\xi^a$ and
$\mu^a_1Dp_1+ \mu^a_2Dp_2$ belong to
$(L^2(\mathbb{R}^3_a))^2\subset(L^1_{\hbox{loc}}(\mathbb{R}^3_a))^2$,
using Fubini theorem with  (\ref{poincare}) and (\ref{idp}) one
entails that
\begin{equation}\nonumber\begin{array}{l}\displaystyle{\int_{\mathbb{R}^3_a}\xi^a\varphi dx=
\int_0^{+\infty}\left(\int_{\mathbb{R}^2}\xi^a\varphi
dx_1dx_2\right)dx_3=}\displaystyle{
\int_0^1\left(\int_{\mathbb{R}^2}\left(\mu^a_1Dp+
\mu^a_2Dq\right)\varphi dx_1dx_2\right)
dx_3=}\\\\\displaystyle{\int_{\mathbb{R}^2\times
]0,1[}\left(\mu^a_1Dp+ \mu^a_2Dq\right)\varphi
dx,\quad\forall\varphi\in
C_0^\infty(\mathbb{R}^3_a)},\end{array}\end{equation}
 that
is
\begin{equation}\label{idxia} \xi^a(x_1,x_2,x_3)=\left\{\begin{array}{ll} (0,0),
\hbox{ a.e. in } \mathbb{R}^2\times ]1,+\infty[, \\\\
\mu^a_1(x_3)Dp(x_1,x_2)+ \mu^a_2(x_3)Dq(x_1,x_2),\hbox{ a.e. in }
\mathbb{R}^2\times ]0,1[,\end{array}\right.\end{equation} with $p$
(resp. $q$) the unique solution of (\ref{resequxx1}) (resp.
(\ref{resequxx2})). Consequently, the first two limits in
(\ref{condu}) hold true for the whole sequence.

The third step is devoted to identify $\xi^b$. To this aim, in
equation (\ref{equariscalata})  with
$\underline{m}=\underline{m}_n$ choose $u^a=c_n$ and
$u^b=\varphi+c_n$, with $\varphi\in C^\infty_0 (\mathbb{R}^3_b)$
and $c_n=-(\vert B^a_n \vert+ \vert B^b_n
\vert)^{-1}\int_{B^b_n}\varphi dx$ (such that $(u^a,u^b)\in {\cal
U}_n$).  By multiplying this equation by $h^2_n$, one has
\begin{equation}\label{equariscalatazlkb}\left\{\begin{array}{l}
 \displaystyle{\int_{\mathbb{R}^3_b}
\left(D_{x_1}u_n^b,D_{x_2}u_n^b,\frac{1}{h^2_n}
D_{x_3}u_n^b\right)\left(h^2_n D_{x_1}\varphi, h^2_n
D_{x_2}\varphi, D_{x_3}\varphi\right)
dx=}\\\\\displaystyle{\int_{\Omega^b}\left(h^2_n D_{x_1}\varphi,
h^2_n D_{x_2}\varphi, D_{x_3}\varphi\right){\underline{m}^b_n}
dx},\quad\forall \varphi\in C^\infty_0
(\mathbb{R}^3_b).\end{array}\right.\end{equation} Then, passing to
the limit, as $n$ diverges, in (\ref{equariscalatazlkb}),
convergences  (\ref{mnconv}), (\ref{conunL60}) and (\ref{condu})
give that
$$\int_{\mathbb{R}^3_b}\xi^bD_{x_3}\varphi dx=\int_{\Omega^b}\mu^b_3D_{x_3}\varphi
dx\quad\forall \varphi\in C^\infty_0 (\mathbb{R}^3_b),$$ which
provides that, for $(x_1,x_2)$ a.e. in $\mathbb{R}^2$, the
function $\xi^b (x_1,x_2,
\cdot)-\widetilde{\mu^b_3}(x_1,x_2,\cdot)$ is constant in
$]-\infty,0[$, where $\widetilde{\mu^b_3}$ denotes the zero
extension of $\mu^b_3$ on $\mathbb{R}^3_b\setminus\Omega^b$. On
the other hand, for $(x_1,x_2)$ a.e. in $\mathbb{R}^2$, $\xi^b
(x_1,x_2, \cdot)-\widetilde{\mu^b_3}(x_1,x_2,\cdot)\in
L^2(]-\infty,0[)$. Then, for $(x_1,x_2)$ a.e. in $\mathbb{R}^2$,
it results that
$$ \xi^b (x_1,x_2, \cdot)=\widetilde{\mu^b_3}(x_1,x_2, \cdot), \quad \hbox{ a.e. in }]-\infty,0[,$$
from which, arguing as above, it follows that
\begin{equation}\label{idxib} \xi^b(x_1,x_2,x_3)=\left\{\begin{array}{ll} 0,
\hbox{ a.e. in } \mathbb{R}^3_b\setminus\Omega^b, \\\\
\mu^b(x_3),\hbox{ a.e. in }
\Omega^b.\end{array}\right.\end{equation} Consequently, also the
last limit in  (\ref{condu}) holds true for the whole sequence.

The last step is devoted to prove that convergences in
(\ref{conunL60}) and (\ref{condu}) are strong, and to obtain
convergence (\ref{conEmag}). By passing to the limit in
(\ref{equariscalata}) with $\underline{m}=\underline{m}_n$,
$u^a=u^a_n$ and $u^b=u^b_n$, and using (\ref{mnconv}),
(\ref{conunL60}), (\ref{condu}), (\ref{idxia}), (\ref{idxib}) and
equation (\ref{17giugno}) with test function $\mu^a_1p+ \mu^a_2q$,
one obtains the convergence of the energies:
\begin{equation}\label{limen}\begin{array}{l}
\displaystyle{\lim_n\Bigg[\int_{\mathbb{R}^3_a}
\left\vert\left(\frac{1}{h_n}D_{x_1}u_n^a,\frac{1}{h_n}D_{x_2}u_n^a,
D_{x_3}u_n^a\right)\right\vert^2dx+}\\\\
\displaystyle{\int_{\mathbb{R}^3_b}\left\vert
\left(D_{x_1}u_n^b,D_{x_2}u_n^b,
\frac{1}{h^2_n}D_{x_3}u_n^b\right)\right\vert^2
dx\Bigg]=}\\\\\displaystyle{ \lim_n\Bigg[\int_{\Omega^a}
\left(\frac{1}{h_n}D_{x_1}u_n^a,\frac{1}{h_n}D_{x_2}u_n^a,
D_{x_3}u_n^a\right)\underline{m}_n^a dx+}\\\\
\displaystyle{\int_{\Omega^b} \left(D_{x_1}u_n^b,D_{x_2}u_n^b,
\frac{1}{h^2_n}D_{x_3}u_n^b\right)\underline{m}_n^b dx\Bigg]=}
\\\\\displaystyle{
\int_{\Omega^a} \left(\mu^a_1Dp+
\mu^a_2Dq\right)(\mu^a_1,\mu^a_2)dx+
\int_{\Omega^b}\vert\mu^b_3\vert^2dx
=}\\\\\displaystyle{\int_{\mathbb{R}^2\times]0,1[}\vert \mu^a_1Dp+
\mu^a_2Dq\vert^2dx + \int_{\Omega^b}\vert\mu^b_3\vert^2dx. }
\end{array}\end{equation}
By combining (\ref{conunL60}), (\ref{condu}), (\ref{idxia}),
(\ref{idxib}) with (\ref{limen}), one deduces limits in
(\ref{conlimprin}). Limit (\ref{conEmag}) is a consequence of
(\ref{conlimprin}) and (\ref{idxiaiiz}).

\end{proof}

\subsection{Proof of theorem \ref{ultimo}\label{abreul}}
\begin{proof} By choosing $\underline{m}=\left((0,1,0),(0,1,0)\right)$ as test function in (\ref{problemariscalato}), and taking into account (\ref{forze})
and that  $\vert \underline{m}_n\vert=1$ a.e. in $\Omega^a\bigcup
\Omega^b$, it is easy to see that there exists $c\in ]0,+\infty[$
such that
\begin{equation}\nonumber\begin{array}{l}\displaystyle{\int_{\Omega^a}\left\vert\left(\frac{1}{h_n}D_{x_1}\underline{m}^a_n|
\frac{1}{h_n}D_{x_2}\underline{m}^a_n|
D_{x_3}\underline{m}^a_n\right)\right\vert^2dx+\int_{\Omega^b}\left\vert\left(D_{x_1}\underline{m}^b_n|D_{x_2}\underline{m}^b_n|\frac{1}{h^2_n}
D_{x_3}\underline{m}^b_n\right)\right\vert^2dx\leq}\\\\
c+E^{mag}_n((0,1,0),(0,1,0)), \quad\forall n\in \mathbb{N},
\end{array}
\end{equation}
where $E^{mag}_n$ is defined (\ref{Emag}). in Consequently, since
proposition \ref{limvarprin} provides that the sequence
$\left\{E^{mag}_n((0,1,0),(0,1,0))\right\}_{n\in \mathbb{N}}$ is
bounded, one obtains that there exists $c\in]0,+\infty[$ such that
\begin{equation}\nonumber\left\{\begin{array}{l}\Vert D_{x_1}\underline{m}^a_n\Vert_{(L^2(\Omega^a))^3}\leq ch_n,\quad \Vert D_{x_2}\underline{m}^a_n\Vert_{(L^2(\Omega^a))^3}\leq ch_n, \quad
\Vert D_{x_3}\underline{m}^a_n\Vert_{(L^2(\Omega^a))^3}\leq c,
\\\\
\Vert D_{x_1}\underline{m}^b_n\Vert_{(L^2(\Omega^b))^3}\leq
c,\quad \Vert
D_{x_2}\underline{m}^b_n\Vert_{(L^2(\Omega^b))^3}\leq c, \quad
\Vert D_{x_3}\underline{m}^b_n\Vert_{(L^2(\Omega^b))^3}\leq
ch^2_n,\end{array}\right. \end{equation} for every $n\in
\mathbb{N}$. Then, taking into account that $\vert
\underline{m}_n\vert=1$ a.e. in $\Omega^a\bigcup \Omega^b$, there
exist an increasing sequence of positive integer numbers
$\{n_i\}_{i\in\mathbb N}$,
$\widehat{\mu}=(\widehat{\mu}^a,\widehat{\mu}^b)\in {\cal M}$,
$\zeta^a\in \left(L^2(\Omega^a, \mathbb{R}^3)\right)^2$ and
$\zeta^b\in L^2(\Omega^b,\mathbb{R}^3)$ such that
\begin{equation}\label{ifconvergenzamweak}\begin{array}{l}
\underline{m}^{a}_{n_i}\rightharpoonup \widehat{\mu}^a\hbox{
weakly  in } H^1(\Omega^a,\mathbb{R}^3),\quad
\underline{m}^{b}_{n_i}\rightharpoonup\widehat{\mu}^b\hbox{ weakly
in } H^1(\Omega^b,\mathbb{R}^3),\end{array}
\end{equation}
\begin{equation}\label{ifconvergenzamweakboz}\left\{\begin{array}{l}\displaystyle{
\left(\frac{1}{h_{n_i}}D_{x_1}\underline{m}^a_{n_i},\frac{1}{h_{n_i}}
D_{x_2}\underline{m}^a_{n_i}\right) \rightharpoonup\zeta^a \hbox{
weakly  in }\left(L^2(\Omega^a,
\mathbb{R}^3)\right)^2,}\\\\\displaystyle{
\frac{1}{h^2_{n_i}}D_{x_3}\underline{m}^b_{n_i}
\rightharpoonup\zeta^b \hbox{ weakly  in }L^2(\Omega^b,
\mathbb{R}^3),}
\end{array}\right.
\end{equation}
as $i$ diverges. Consequently, by virtue of proposition
\ref{limvarprin}, limits in   (\ref{limcon}) hold true and it
results that
\begin{equation}\label{conEmagni}\begin{array}{l}\displaystyle{\lim_iE_{n_i}^{mag}(\underline{m}_{n_i})=
\frac{1}{2}\Bigg(\alpha(\Theta)\int_0^1\vert\widehat{\mu}^a_1\vert^2dx_3+\beta(\Theta)\int_0^1\vert\widehat{\mu}^a_2\vert^2dx_3+}\\\\
\displaystyle{
\gamma(\Theta)\int_0^1\widehat{\mu}^a_1\widehat{\mu}^a_2dx_3+\int_{\Theta}\vert
\widehat{\mu}_3^b\vert^2dx_1dx_2\Bigg),}
\end{array}\end{equation}
where $\alpha(\Theta)$, $\beta(\Theta)$ and $\gamma(\Theta)$ are
defined by (\ref{abc}) with $S=\Theta$.

Now, the goal is to identify $\widehat{\mu}$, $\zeta^a$,
$\zeta^b$, to obtain strong convergences in
(\ref{ifconvergenzamweak}) and in (\ref{ifconvergenzamweakboz}),
and to prove convergence (\ref{convenergie}). To this aim, for
$(\mu^a,\mu^b)\in {\cal M}_{\hbox{reg}}=\{(\mu^a,\mu^b)\in
C^1([0,1],S^2)\times C^1(\overline\Theta,S^2):$
$\mu^a(0)=\mu^b(0)\}$ let, for every $n\in \mathbb{N}$, $v_n
=(v_n^a,v_n^b)\in {\cal M}_n$ be the couple of functions defined
in (2.37) of \cite{GaH1} with $w=\mu^a$ and $\zeta=\mu^b$. Then,
in \cite{GaH1} it is proved that
\begin{equation}\label{rec}\begin{array}{l}\displaystyle{\lim_n\Bigg[\int_{\Omega^a}
\left(\lambda\left\vert\left(\frac{1}{h_n}D_{x_1}v^a_n|\frac{1}{h_n}D_{x_2}v^a_n|
D_{x_3}v^a_n\right)\right\vert^2-2f^a_n v^a_n
 \right)dx+}\\\\
\displaystyle{
\int_{\Omega^b}\left(\lambda\left\vert\left(D_{x_1}v^b_n|D_{x_2}v^b_n|\frac{1}{h^2_n}
D_{x_3}v^b_n\right)\right\vert^2-2f^b_nv^b_n\right)dx\Bigg]=}
\\\\\displaystyle{\vert \Theta\vert\int_0^1\left(\lambda\left\vert
\frac{d\mu^a}{dx_3}\right\vert^2-2F^a\mu^a
 \right)dx_3+\int_{\Theta}\left(\lambda\left\vert
D\mu^b\right\vert^2-2F^b\mu^b\right)dx_1dx_2.}
\end{array}\end{equation}
Moreover, it is easy to see that
\begin{equation}\label{lim2rec} v_n^a\rightarrow \mu^a \hbox{ strongly in }L^2(\Omega^a,\mathbb{R}^3),\quad\quad v^b_n\rightarrow
\mu^b\hbox{ strongly in }
L^2(\Omega^b,\mathbb{R}^3),\end{equation}
 as $n$ diverges. Then,  it follows from (\ref{rec}), (\ref{lim2rec}) and
 proposition \ref{limvarprin} that
\begin{equation}\nonumber\lim_n E_n(v_n)=E(\mu^a,\mu^b)\end{equation}
from which,  using l.s.c. arguments, (\ref{forze}),
(\ref{ifconvergenzamweak}), (\ref{ifconvergenzamweakboz}) and
(\ref{conEmagni}), one obtains that
\begin{equation}\label{lim2recvnzz}\left\{\begin{array}{l}\displaystyle{\lambda\int_{\Omega^a}\vert
\zeta^a\vert^2dx +\lambda\int_{\Omega^b}\vert \zeta^b\vert^2dx+
E(\widehat{\mu}^a,\widehat{\mu}^b)\leq\liminf_i
E_{n_i}(\underline{m}_{n_i})\leq}\\\\\displaystyle{\limsup_i
E_{n_i}(\underline{m}_{n_i})\leq\lim_i
E_{n_i}(v_{n_i})=E(\mu^a,\mu^b).}\end{array}\right.\end{equation}
Since (\ref{lim2recvnzz}) holds true for every $(\mu^a,\mu^b)\in
{\cal M}_{\hbox{reg}}$ and ${\cal M}_{\hbox{reg}}$ is dense in
${\cal M}$ (see \cite{GaH1}), one has that (\ref{lim2recvnzz})
holds also true  for every $(\mu^a,\mu^b)\in {\cal M}$.
Consequently, $\zeta^a=0$, $\zeta^b=0$,
$(\widehat{\mu}^a,\widehat{\mu}^b)$ solves (\ref{problemafinale})
and limit (\ref{convenergie}) holds true. Finally, combining
(\ref{convenergie}) with (\ref{forze}),
(\ref{ifconvergenzamweak}), (\ref{ifconvergenzamweakboz}) and
(\ref{conEmagni}) one obtains that convergences in
(\ref{ifconvergenzamweak}) and in (\ref{ifconvergenzamweakboz})
are strong.
\end{proof}

\section{Wire - wire\label{w-w}}

This section is devoted to study the asymptotic behavior, as $n$
diverges, of problem (\ref{proiniz}) in the second case, that is
the case  wire - wire. Specifically, for every  $n \in \mathbb N$,
let $\Omega_n^a=]-h_n,0[^2\times ]0,1[$,
$\Omega_n^{b,l}=]0,1[\times ]-h_n,0[^2$ and $\Omega_n^{b,r}=
]-h_n,0]^3$. Then, we study the asymptotic behavior, as $n$
diverges, of problem (\ref{minoriginale}) with
$\Omega_n=\Omega_n^a\cup\Omega_n^{b,l}\cup\Omega_n^{b,r}$ (see
Fig. 2).

\subsection{The rescaled problem }
By setting
\begin{equation}\nonumber\left\{\begin{array}{ll}
\mathbb{R}^3_a=\{(x_1,x_2,x_3)\in \mathbb{R}^3: x_3>0\},\\\\
 \mathbb{R}^3_{b,l}=\{(x_1,x_2,x_3)\in \mathbb{R}^3:
 x_3<0,\,x_1>0\},\\\\ \mathbb{R}^3_{b,r}=\{(x_1,x_2,x_3)\in \mathbb{R}^3:
 x_3< 0,\,x_1<0\},\end{array}\right.\end{equation}
for every $n\in \mathbb{N}$, problem (\ref{minoriginale}) is
reformulated
 on a fixed domain through the following rescaling
\begin{equation}\label{resca}T_n:(x_1,x_2,x_3)\in \mathbb{R}^3\rightarrow T_n(x_1,x_2,x_3)=\left\{\begin{array}{ll}(h_nx_1,h_nx_2,x_3), \hbox{ if }(x_1,x_2,x_3)\in
\mathbb{R}^3_a,\\\\(x_1,h_nx_2,h_nx_3), \hbox{ if }
(x_1,x_2,x_3)\in \mathbb{R}^3_{b,l},\\\\
(h_nx_1,h_nx_2,h_nx_3),\hbox{ if } (x_1,x_2,x_3)\in
\mathbb{R}^3_{b,r}.
\end{array}\right.\end{equation}
Namely, setting $$\Omega^a=]-1,0[^2 \times ]0,1[,\quad
\Omega^{b,l}=]0,1[\times ]-1,0[^2,\quad\Omega^{b,r}=]-1,0[^3,$$
and
$$B^a_n=\left]-\frac{2 }{h_n}, \frac{2 }{h_n}\right[^2\times ]0,2[,\quad B^{b,l}_n=
]0, 2[\times \left]-\frac{2}{h_n},0\right[^2,\quad B^{b,r}_n=
\left]-\frac{2}{h_n},0\right[^3,\quad n\in \mathbb{N},$$ the
    space ${\cal U}$ defined in
(\ref{calU}) is rescaled in the following
\begin{equation}\label{spazioriscalatoUw-w}\begin{array}{ll} {\cal U}_n= \big\{&u=(u^a, u^{b,l},u^{b,r})
\in L^1_{loc} (\overline{\mathbb{R}^3_a})\times L^1_{loc}
(\overline{\mathbb{R}^3_{b,l}})\times L^1_{loc}
(\overline{\mathbb{R}^3_{b,r}})\,\,:\\\\& (u^a_{|_{B^a_n}},
u^{b,l}_{|_{B^{b,l}_n}},u^{b,r}_{|_{B^{b,r}_n}}) \in L^2
(B^a_n)\times L^2 (B^{b,l}_n)\times L^2 (B^{b,r}_n),\\\\&
(Du^a,Du^{b,l},Du^{b,r})\in
(L^2(\mathbb{R}^3_a))^3\times (L^2(\mathbb{R}^3_{b,l}))^3\times (L^2(\mathbb{R}^3_{b,r}))^3,\\\\
&\quad\displaystyle{\int_{B^a_n}u^adx+\int_{B^{b,l}_n}u^{b,l}dx+h_n\int_{B^{b,r}_n}u^{b,r}dx=0,}\\\\&
u^a(x_1,x_2,0)= u^{b,l}(h_nx_1,x_2,0),\hbox{ for }(x_1,x_2) \hbox{
a.e. in }]0,+\infty[\times\mathbb{R},\\\\
&u^a(x_1,x_2,0)= u^{b,r}(x_1,x_2,0),\hbox{ for }(x_1,x_2)\hbox{
a.e. in }]-\infty,0[\times\mathbb{R}, \\\\
&u^{b,l}(0,x_2,x_3)= u^{b,r}(0,x_2,x_3),\hbox{ for
}(x_2,x_3)\hbox{ a.e. in }\mathbb{R}\times]-\infty,0[
\big\}.\end{array}
\end{equation}  Then,  for every
$\underline{m}=(\underline{m}^a,\underline{m}^{b,l},\underline{m}^{b,r})\in
L^2(\Omega^a,\mathbb{R}^3)\times
L^2(\Omega^{b,l},\mathbb{R}^3)\times
L^2(\Omega^{b,r},\mathbb{R}^3)$, the following equation
\begin{equation}\label{equariscalataw-w}\left\{\begin{array}{l}u_{\underline{m},n}=(u_{\underline{m},n}^a,u_{\underline{m},n}^{b,l}, u_{\underline{m},n}^{b,r})
\in {\cal U}_n,\\\\
 \displaystyle{\int_{\mathbb{R}^3_a}
\left(\frac{1}{h_n}D_{x_1}u_{\underline{m},n}^a,\frac{1}{h_n}D_{x_2}u_{\underline{m},n}^a,
D_{x_3}u_{\underline{m},n}^a\right)\left(\frac{1}{h_n}D_{x_1}u^a,\frac{1}{h_n}D_{x_2}u^a,
D_{x_3}u^a\right) dx+}\\\\
 \displaystyle{\int_{\mathbb{R}^3_{b,l}}
\left(D_{x_1}u_{\underline{m},n}^{b,l},\frac{1}{h_n}D_{x_2}u_{\underline{m},n}^{b,l},
\frac{1}{h_n}D_{x_3}u_{\underline{m},n}^{b,l}\right)\left(D_{x_1}u^{b,l},\frac{1}{h_n}D_{x_2}u^{b,l},
\frac{1}{h_n}D_{x_3}u^{b,l}\right)
dx+}\\\\
\displaystyle{\frac{1}{h_n}\int_{\mathbb{R}^3_{b,r}}
\left(D_{x_1}u_{\underline{m},n}^{b,r},D_{x_2}u_{\underline{m},n}^{b,r},
D_{x_3}u_{\underline{m},n}^{b,r}\right)\left(D_{x_1}u^{b,r},D_{x_2}u^{b,r},
D_{x_3}u^{b,r}\right)
dx=}\\\\\displaystyle{\int_{\Omega^a}\left(\frac{1}{h_n}D_{x_1}u^a,\frac{1}{h_n}D_{x_2}u^a,
D_{x_3}u^a\right){\underline{m}^a} dx+}\\\\
\displaystyle{\int_{\Omega^{b,l}}
\left(D_{x_1}u^{b,l},\frac{1}{h_n}D_{x_2}u^{b,l},
\frac{1}{h_n}D_{x_3}u^{b,l}\right){\underline{m}^{b,l}}
dx+}\\\\\displaystyle{\int_{\Omega^{b,r}}
\left(D_{x_1}u^{b,r},D_{x_2}u^{b,r},
D_{x_3}u^{b,r}\right){\underline{m}^{b,r}} dx,}\quad\forall
u=(u^a,u^{b,l},u^{b,r})\in {\cal
U}_n,\end{array}\right.\end{equation}
 which  rescales  equation (\ref{resequ}), admits a unique
solution. We note that
$u_{\underline{m},n}=(u_{\underline{m},n}^a,u_{\underline{m},n}^{b,l},
u_{\underline{m},n}^{b,r})$ belongs to $ H^1(\mathbb{R}^3_a)\times
H^1(\mathbb{R}^3_{b,l})\times H^1(\mathbb{R}^3_{b,r})$ up to an
additive constant.

For every $n\in \mathbb{N}$, $H^1(\Omega_n,S^2)$, $F_n\in
L^2(\Omega_n)$ and the functional involved in problem
(\ref{minoriginale}) with
$\Omega_n=\Omega_n^a\cup\Omega_n^{b,l}\cup\Omega_n^{b,r}$ and
renormalized by $h^2_n$ are rescaled in
\begin{equation}\label{spazioriscalatoMw-w}\begin{array}{ll} {\cal M}_n= \Big\{&\underline{m}=(\underline{m}^a, \underline{m}^{b,l}, \underline{m}^{b,r})
\in H^1 (\Omega^a,S^2)\times H^1
(\Omega^{b,l},S^2)\times H^1
(\Omega^{b,r},S^2)\,:\\\\
 & \underline{m}^a(x_1,x_2,0)= \underline{m}^{b,r}(x_1,x_2,0),\hbox{ for }(x_1,x_2) \hbox{
a.e. in }]-1,0[^2,\\\\ & \underline{m}^{b,l}(0,x_2,x_3)=
\underline{m}^{b,r}(0,x_2,x_3),\hbox{ for }(x_2,x_3) \hbox{ a.e.
in }]-1,0[^2\Big\},\end{array}
\end{equation}
\begin{equation}\label{riscalamentoforzew-w}\begin{array}{ll}
f_n:x\in \Omega^a\cup \Omega^{b,l}\cup\Omega^{b,r}\longrightarrow\\\\
f_n(x)=\left\{
\begin{array}{ll} f^{a}_n(x)=
 F_n(h_nx_1,h_nx_2,x_3),\quad \hbox{for }x\hbox{ a.e. in
}\Omega^a, \\\\
 f^{b,l}_n(x)= F_n(x_1,h_nx_2, h_nx_3), \quad
\hbox{for }x\hbox{ a.e. in }\Omega^{b,l},
 \\\\
 f^{b,r}_n(x)= F_n(h_nx_1,h_nx_2, h_nx_3), \quad
\hbox{for }x\hbox{ a.e. in }\Omega^{b,r},
\end{array}\right.\end{array}
\end{equation}
and
\begin{equation}\label{funzionaleriscalatoEw-w}\begin{array}{l}\displaystyle{E_n:\underline{m}=(\underline{m}^a,\underline{m}^{b,l}, \underline{m}^{b,r})\in
{\cal
M}_n\longrightarrow}\\\\\displaystyle{\int_{\Omega^a}\left(\lambda\left\vert\left(\frac{1}{h_n}D_{x_1}\underline{m}^a|
\frac{1}{h_n}D_{x_2}\underline{m}^a|
D_{x_3}\underline{m}^a\right)\right\vert^2+\varphi(\underline{m}^a)-2f^a_n\underline{m}^a
 \right)dx+}\\\\\displaystyle{\frac{1}{2}\int_{\Omega^a}\left(
\left(\frac{1}{h_n}D_{x_1}u^a_{\underline{m},n},\frac{1}{h_n}D_{x_2}u^a_{\underline{m},n},
D_{x_3}u^a_{\underline{m},n}\right)\underline{m}^a
 \right)dx+}\\\\
\displaystyle{
\int_{\Omega^{b,l}}\left(\lambda\left\vert\left(D_{x_1}\underline{m}^{b,l}|\frac{1}{h_n}D_{x_2}\underline{m}^{b,l}|\frac{1}{h_n}
D_{x_3}\underline{m}^{b,l}\right)\right\vert^2+\varphi(\underline{m}^{b,l})-2f^{b,l}_n\underline{m}^{b,l}\right)dx+}
\\\\
\displaystyle{\frac{1}{2}\int_{\Omega^{b,l}}\left(
\left(D_{x_1}u^{b,l}_{\underline{m},n},\frac{1}{h_n}D_{x_2}u^{b,l}_{\underline{m},n},
\frac{1}{h_n}D_{x_3}u^{b,l}_{\underline{m},n}\right)\underline{m}^{b,l}\right)dx+}\\\\
\displaystyle{h_n
\int_{\Omega^{b,r}}\left(\lambda\left\vert\left(\frac{1}{h_n}D_{x_1}\underline{m}^{b,r}|\frac{1}{h_n}D_{x_2}\underline{m}^{b,r}|\frac{1}{h_n}
D_{x_3}\underline{m}^{b,r}\right)\right\vert^2+\varphi(\underline{m}^{b,r})-2f^{b,r}_n\underline{m}^{b,r}\right)dx+}
\\\\
\displaystyle{\frac{1}{2}\int_{\Omega^{b,r}}\left(
\left(D_{x_1}u^{b,r}_{\underline{m},n},D_{x_2}u^{b,r}_{\underline{m},n},
D_{x_3}u^{b,r}_{\underline{m},n}\right)\underline{m}^{b,r}\right)dx,}
\end{array}\end{equation}
respectively. Then,  the function defined by
\begin{equation}\nonumber\left\{
\begin{array}{ll}
 M_n(h_nx_1,h_nx_2,x_3),\hbox{ for }x\hbox{ a.e. in
}\Omega^a, \\\\
 M_n(x_1,h_nx_2,h_n x_3),
\hbox{ for }x\hbox{ a.e. in }\Omega^{b,l},\\\\
 M_n(h_nx_1,h_nx_2,h_n x_3),
\hbox{ for }x\hbox{ a.e. in }\Omega^{b,r},
\end{array}\right.
\end{equation}
with $M_n$ solution of problem (\ref{minoriginale}) with
$\Omega_n=\Omega_n^a\cup\Omega_n^{b,l}\cup\Omega_n^{b,r}$, is a
minimizer of the following problem:
\begin{equation}\label{problemariscalatow-w}\begin{array}{l}
\min\left \{E_n(\underline{m}): \underline{m}\in {\cal M}_n
\right\}.
\end{array}\end{equation}
Actually, the goal of this  section becomes to study the
asymptotic behavior, as $n$ diverges, of problem
(\ref{problemariscalatow-w}). To this aim, it will be assumed that
\begin{equation}\label{forzew-w}\left\{ \begin{array}{l} f^{a}_n\rightharpoonup f^{a}\hbox{ weakly
in } L^2(\Omega^a, \mathbb R^3),\\\\ f^{b,l}_n\rightharpoonup
f^{b,l}\hbox{ weakly in } L^2(\Omega^{b,l},\mathbb R^3),\\\\
f^{b,r}_n\rightharpoonup f^{b,r}\hbox{ weakly in }
L^2(\Omega^{b,r},\mathbb R^3).\end{array}\right.\end{equation}

\subsection{The main result }
Let
\begin{equation}\label{limitspacew-w}\begin{array}{ll}{\cal M}=\Big\{\mu=(\mu^a,\mu^{b,l})\in H^1(\Omega^a,S^2)\times
H^1(\Omega^{b,l},S^2)\,:\,
\mu^a\hbox{ is independent of } (x_1,x_2),\\\\
\quad\quad\mu^b\hbox{ is independent of } (x_2,x_3),\quad
\mu^a(0)=\mu^{b,l}(0) \Big\}\simeq
\\\\ \Big\{\mu=(\mu^a,\mu^{b,l})\in H^1(]0,1[,S^2)\times
H^1(]0,1[,S^2)\,:\, \mu^a(0)=\mu^{b,l}(0) \Big\},
\end{array}\end{equation}
\begin{equation}\label{averagew-w}\left\{\begin{array}{ll}F^a: x_3\in ]0,1[\longrightarrow{\displaystyle\int_{-1}^{0}\int_{-1}^{0}
f^a(x_1,x_2,x_3)dx_1dx_2},\\\\
F^{b,l}:x_1\in ]0,1[\longrightarrow
{\displaystyle\int_{-1}^{0}\int_{-1}^{0}
f^{b,l}(x_1,x_2,x_3)dx_2dx_3},
\end{array}\right.\end{equation}
 and
\begin{equation}\label{funlimitew-w}\begin{array}{l}\displaystyle{E:\mu=(\mu^a,\mu^{b,l})=((\mu^a_1,\mu^a_2,\mu^a_3),(\mu^{b,l}_1,\mu^{b,l}_2,\mu^{b,l}_3))\in
{\cal
M}\longrightarrow}\\\\\displaystyle{\int_0^1\left(\lambda\left\vert
\frac{d\mu^a}{dx_3}\right\vert^2+\varphi(\mu^a)-2F^a\mu^a
 \right)dx_3+}\\\\\displaystyle{
\frac{1}{2}\Bigg(\alpha(]-1,0[^2)\int_0^1\vert\mu^a_1\vert^2dx_3+\beta(]-1,0[^2)\int_0^1\vert\mu^a_2\vert^2dx_3+\gamma(]-1,0[^2)\int_0^1\mu^a_1\mu^a_2dx_3\Bigg)+}\\\\
\displaystyle{\int_0^1\left(\lambda\left\vert
\frac{d\mu^{b,l}}{dx_1}\right\vert^2+\varphi(\mu^{b,l})-2F^{b,l}\mu^{b,l}
 \right)dx_1+}\\\\\displaystyle{
\frac{1}{2}\Bigg(\alpha(]-1,0[^2)\int_0^1\vert\mu^{b,l}_2\vert^2dx_1+\beta(]-1,0[^2)\int_0^1\vert\mu^{b,l}_3\vert^2dx_1+\gamma(]-1,0[^2)\int_0^1
\mu^{b,l}_2\mu^{b,l}_3dx_1\Bigg)}
\end{array}\end{equation}
where  $\alpha(]-1,0[^2)$, $\beta(]-1,0[^2)$ and
$\gamma(]-1,0[^2)$ are defined by (\ref{abc}) with $S=]-1,0[^2$.

This section is devoted to prove the following main result
\begin{Theorem}\label{ultimow-w}Assume  (\ref{forzew-w}).
For every $n\in \mathbb{N}$, let
$\underline{m}_n=(\underline{m}_n^a,\underline{m}_n^{b,l},\underline{m}_n^{b,r})$
be a solution of problem (\ref{problemariscalatow-w}) and
$u_n=(u_n^a,u_n^{b,l},u_n^{b,r})$ be the unique solution of
(\ref{equariscalataw-w}) corresponding to $\underline{m}_n$.
Moreover, let ${\cal M}$ and $E$ be defined by
(\ref{limitspacew-w}) and (\ref{funlimitew-w}), respectively.
Then, there exist an increasing sequence of positive integer
numbers $\{n_i\}_{i\in\mathbb N}$ and
$\widehat{\mu}=(\widehat{\mu}^a,\widehat{\mu}^{b,l})=((\widehat{\mu}^a_1,\widehat{\mu}^a_2,\widehat{\mu}^a_3),
(\widehat{\mu}^{b,l}_1,\widehat{\mu}^{b,l}_2,\widehat{\mu}^{b,l}_3))\in{\cal
M}$, depending on the selected subsequence, such that
\begin{equation}\label{convergenzastrongw-w}\left\{\begin{array}{l}
\underline{m}^{a}_{n_i}\rightarrow \widehat{\mu}^a\hbox{ stongly
in } H^1(\Omega^a,S^2),\\\\
\underline{m}^{b,l}_{n_i}\rightarrow\widehat{\mu}^{b,l}\hbox{
strongly in } H^1(\Omega^{b,l},S^2),\\\\
\underline{m}^{b,r}_{n_i}\rightarrow\widehat{\mu}^a(0)=\widehat{\mu}^{b,l}(0)\hbox{
strongly in } H^1(\Omega^{b,r},S^2),\end{array}\right.
\end{equation}
\begin{equation}\label{convergenzaDstrongw-w}\left\{\begin{array}{ll}
\dfrac{1}{h_{n}}D_{x_1}\underline{m}^{a}_{n}\rightarrow0,\quad\dfrac{1}{h_{n}}D_{x_2}\underline{m}^{a}_{n}\rightarrow
0\hbox{ stongly in }
L^2(\Omega^a,\mathbb{R}^3),\\\\
\dfrac{1}{h_{n}}D_{x_2}\underline{m}^{b,l}_{n}\rightarrow0,\quad\dfrac{1}{h_{n}}D_{x_3}\underline{m}^{b,l}_{n}\rightarrow
0\hbox{ stongly in }
L^2(\Omega^{b,l},\mathbb{R}^3),\\\\
\dfrac{1}{\sqrt{h_{n}}}D\underline{m}^{b,r}_{n}\rightarrow0\hbox{
stongly in }
\left(L^2(\Omega^{b,r},\mathbb{R}^3)\right)^3,\end{array}\right.
\end{equation}
\begin{equation}\label{limconw-w}\left\{\begin{array}
{llll}
 \dfrac{1}{h_{n_i}}D_{x_1}u_{n_i}^a\rightharpoonup
\xi^a_1,
 \quad  \dfrac{1}{h_{n_i}}D_{x_2}u_{n_i}^a\rightharpoonup\xi^a_2, \quad D_{x_3}u_n^a\rightharpoonup0\quad\hbox{ weakly in
 }L^2(\mathbb{R}^3_a),\\\\D_{x_1}u_n^{b,l}\rightharpoonup0,
 \quad  \dfrac{1}{h_{n_i}}D_{x_2}u_{n_i}^{b,l}\rightharpoonup\xi^{b,l}_2, \quad \dfrac{1}{h_{n_i}}D_{x_3}u_{n_i}^{b,l}\rightharpoonup\xi^{b,l}_3\quad\hbox{ weakly in
 }L^2(\mathbb{R}^3_{b,l}),\\\\
Du^{b,r}_n \rightarrow0\quad\hbox{ strongly in
 }\left(L^2(\mathbb{R}^3_{b,r})\right)^3,
\end{array}\right.\end{equation}
as $n$ and $i$ diverge, where $\widehat{\mu}$ is a solution of the
following problem:
\begin{equation}\label{problemafinalew-w}\begin{array}{l}
E(\widehat{\mu})=\min\left \{E(\mu): \mu\in {\cal M} \right\},
\end{array}\end{equation} and
\begin{equation}\label{idxiaiizw-w} (\xi_1^a,\xi^a_2)(x_1,x_2,x_3)=
\left\{\begin{array}{ll} (0,0),
\hbox{ a.e. in } \mathbb{R}^2\times ]1,+\infty[, \\\\
\mu^a_1(x_3)Dp(x_1,x_2)+ \mu^a_3(x_3)Dq(x_1,x_2),\hbox{ a.e. in }
\mathbb{R}^2\times ]0,1[,\end{array}\right.\end{equation}
\begin{equation}\label{"ìidxiaiizw-w} (\xi_2^{b,l},\xi^{b,l}_3)(x_1,x_2,x_3)=
\left\{\begin{array}{ll} (0,0),
\hbox{ a.e. in } ]1,+\infty[\times \mathbb{R}\times]-\infty,0[, \\\\
\mu^{b,l}_2(x_1)Dp(x_2,x_3)+ \mu^{b,l}_3(x_1)Dq(x_2,x_3),\hbox{
a.e. in } ]0,1[\times
\mathbb{R}\times]-\infty,0[,\end{array}\right.\end{equation}with
$p$ (resp. $q$) the unique solution of (\ref{resequxx1}) (resp.
(\ref{resequxx2})). Moreover,
      the convergence of the energies holds true, i.e.
\begin{equation}\label{convenergiew-w}\begin{array}{l}\displaystyle{
\lim_nE_n(\underline{m}_n)= E(\widehat{\mu}).}
\end{array}\end{equation}
\end{Theorem}

\subsection{A convergence result for  the magnetostatic
energy \label{magpotw-w}}
\begin{Proposition}\label{limvarprinw-w}
Let
$\{\underline{m}_n=(\underline{m}_n^a,\underline{m}_n^{b,l},\underline{m}_n^{b,r}
)\}_{n\in\mathbb{ N}}\subset L^2(\Omega^a,S^2)\times
L^2(\Omega^{b,l},S^2)\times L^2(\Omega^{b,r},S^2)$, and let
 $\mu^a=(\mu^a_1,\mu^a_2,\mu^a_3)\in L^2(\Omega^a,S^2)$ be  independent
of $(x_1,x_2)$   and
$\mu^{b,l}=(\mu^{b,l}_1,\mu^{b,l}_2,\mu^{b,l}_3)\in
L^2(\Omega^{b,l},S^2)$ be  independent of $(x_2,x_3)$  such that
\begin{equation}\label{mnconvw-w}\left\{\begin{array}{l}\underline{m}_n^a\rightarrow
\mu^a \hbox{ strongly in }L^2(\Omega^a,\mathbb{R}^3)
,\\\\
\underline{m}_n^{b,l}\rightarrow \mu^{b,l} \hbox{ strongly in
}L^2(\Omega^{b,l},\mathbb{R}^3),\end{array}\right.\end{equation}
as $n$ diverges. Moreover, for every $n\in\mathbb{ N}$ let
$u_n=(u_n^a,u_n^{b,l},u_n^{b,r})$ be the unique solution of
(\ref{equariscalataw-w}) corresponding to $\underline{m}_n$. Then,
it results that
\begin{equation}\label{conlimprinw-w}\left\{\begin{array}
{llll}
 \dfrac{1}{h_n}D_{x_1}u_n^a\rightharpoonup
\xi^a_1,
 \quad  \dfrac{1}{h_n}D_{x_2}u_n^a\rightharpoonup\xi^a_2, \quad D_{x_3}u_n^a\rightharpoonup0\quad\hbox{ weakly in
 }L^2(\mathbb{R}^3_a),\\\\D_{x_1}u_n^{b,l}\rightharpoonup0,
 \quad  \dfrac{1}{h_n}D_{x_2}u_n^{b,l}\rightharpoonup\xi^{b,l}_2, \quad \dfrac{1}{h_n}D_{x_3}u_n^{b,l}\rightharpoonup\xi^{b,l}_3\quad\hbox{ weakly in
 }L^2(\mathbb{R}^3_{b,l}),\\\\
Du^{b,r}_n \rightarrow0\quad\hbox{ strongly in
 }\left(L^2(\mathbb{R}^3_{b,r})\right)^3,
\end{array}\right.\end{equation}
as $n$ diverges, where
\begin{equation}\label{idxiaiizw-w} (\xi_1^a,\xi^a_2)(x_1,x_2,x_3)=
\left\{\begin{array}{ll} (0,0),
\hbox{ a.e. in } \mathbb{R}^2\times ]1,+\infty[, \\\\
\mu^a_1(x_3)Dp(x_1,x_2)+ \mu^a_3(x_3)Dq(x_1,x_2),\hbox{ a.e. in }
\mathbb{R}^2\times ]0,1[,\end{array}\right.\end{equation}
\begin{equation}\label{"ìidxiaiizw-w} (\xi_2^{b,l},\xi^{b,l}_3)(x_1,x_2,x_3)=
\left\{\begin{array}{ll} (0,0),
\hbox{ a.e. in } ]1,+\infty[\times \mathbb{R}\times]-\infty,0[, \\\\
\mu^{b,l}_2(x_1)Dp(x_2,x_3)+ \mu^{b,l}_3(x_1)Dq(x_2,x_3),\hbox{
a.e. in } ]0,1[\times
\mathbb{R}\times]-\infty,0[,\end{array}\right.\end{equation}with
$p$ (resp. $q$) the unique solution of (\ref{resequxx1}) (resp.
(\ref{resequxx2})). Furthermore, one has that
\begin{equation}\label{limenw-w}\begin{array}{l}
\displaystyle{ \lim_n\Bigg[\int_{\Omega^a}
\left(\frac{1}{h_n}D_{x_1}u_n^a,\frac{1}{h_n}D_{x_2}u_n^a,
D_{x_3}u_n^a\right)\underline{m}_n^a
dx+}\\\\\displaystyle{\int_{\Omega^{b,l}} \left(D_{x_1}u_n^{b,l},
\frac{1}{h_n}D_{x_2}u_n^{b,l},
\frac{1}{h_n}D_{x_3}u_n^{b,l}\right)\underline{m}_n^{b,l}
dx+}\\\\\displaystyle{\int_{\Omega^{b,r}} \left(D_{x_1}u_n^{b,r},
D_{x_2}u_n^{b,r}, D_{x_3}u_n^{b,r}\right)\underline{m}_n^{b,r}
dx\Bigg]=}
\\\\\displaystyle{
\alpha(]-1,0[^2)\int_0^1\vert\mu^a_1\vert^2dx_3+\beta(]-1,0[^2)\int_0^1\vert\mu^a_2\vert^2dx_3+
\gamma(]-1,0[^2)\int_0^1\mu^a_1\mu^a_2dx_3+}\\\\\displaystyle{
\alpha(]-1,0[^2)\int_0^1\vert\mu^{b,l}_2\vert^2dx_1+\beta(]-1,0[^2)\int_0^1\vert\mu^{b,l}_3\vert^2dx_1+
\gamma(]-1,0[^2)\int_0^1\mu^{b,l}_2\mu^{b,l}_3dx_1,}
\end{array}\end{equation}
where  $\alpha(]-1,0[^2)$, $\beta(]-1,0[^2)$ and
$\gamma(]-1,0[^2)$ are defined by (\ref{abc}) with $S=]-1,0[^2$.
\end{Proposition}
\begin{proof}
By choosing $u=u_n$ as test function in (\ref{equariscalataw-w})
and taking into account  that
$\{(\underline{m}_n^a,\underline{m}_n^{b,l},\underline{m}_n^{b,r}
)\}_{n\in\mathbb{ N}}\subset L^2(\Omega^a,S^2)\times
L^2(\Omega^{b,l},S^2)\times L^2(\Omega^{b,r},S^2)$, there exists
$c\in]0,+\infty[$ such that
\begin{equation}\label{stimew-w}\left\{\begin{array}{l}\left\Vert
\left(\dfrac{1}{h_n}D_{x_1}u_n^a,\dfrac{1}{h_n}D_{x_2}u_n^a,
D_{x_3}u_n^a\right)\right\Vert_{(L^2(R^3_a))^9}\leq
c,\\\\\left\Vert
\left(D_{x_1}u_n^{b,l},\dfrac{1}{h_n}D_{x_2}u_n^{b,l},
\dfrac{1}{h_n}D_{x_3}u_n^{b,l}\right)\right\Vert_{(L^2(R^3_{b,l}))^9}\leq c,\\\\
\dfrac{1}{\sqrt{h_n}}\left\Vert
\left(D_{x_1}u_n^{b,r},D_{x_2}u_n^{b,r},
D_{x_3}u_n^{b,r}\right)\right\Vert_{(L^2(R^3_{b,r}))^9}\leq c,
\end{array}\right.\end{equation}for every $n\in \mathbb{N}$.

 The
last estimate in (\ref{stimew-w}) gives the last limit in
(\ref{conlimprinw-w}).

By arguing as in the first part of the proof of proposition 5.1 in
\cite{GaHa3}, from  the first two estimates in (\ref{stimew-w})
one derives  the third and the fourth limit in
(\ref{conlimprinw-w}).

 By arguing as in the first two steps of the
proof of proposition \ref{limvarprin}, from the first limit in
(\ref{mnconvw-w}) and the first estimate in (\ref{stimew-w}) one
obtains the first two limits in (\ref{conlimprinw-w}) with
$(\xi^a_1, \xi^a_2)$ defined in (\ref{idxiaiizw-w}). Finally,
using the first limit in (\ref{mnconvw-w}), the first three limits
in (\ref{conlimprinw-w}) and also the last one, taking into
account  that $\{(\underline{m}_n^{b,r} )\}_{n\in\mathbb{
N}}\subset L^2(\Omega^{b,r},S^2)$, and using equation
(\ref{17giugno}) with test function $\mu^a_1p+ \mu^a_2q$, one
obtains that
\begin{equation}\label{limenw-wprimo}\left\{\begin{array}{l}
\displaystyle{ \lim_n\int_{\Omega^a}
\left(\frac{1}{h_n}D_{x_1}u_n^a,\frac{1}{h_n}D_{x_2}u_n^a,
D_{x_3}u_n^a\right)\underline{m}_n^a dx=}
\\\\\displaystyle{
\alpha(]-1,0[^2)\int_0^1\vert\mu^a_1\vert^2dx_3+\beta(]-1,0[^2)\int_0^1\vert\mu^a_2\vert^2dx_3+
\gamma(]-1,0[^2)\int_0^1\mu^a_1\mu^a_2dx_3,}
\\\\
\displaystyle{ \lim_n\int_{\Omega^{b,r}} \left(D_{x_1}u_n^{b,r},
D_{x_2}u_n^{b,r}, D_{x_3}u_n^{b,r}\right)\underline{m}_n^{b,r}
dx=0.}
\end{array}\right.\end{equation}

To prove the  fifth and the sixth  limit in (\ref{conlimprinw-w}),
we introduce other rescalings. Specifically, by setting
\begin{equation}\nonumber\left\{\begin{array}{ll}
\mathbb{R}^3_{a,r}=\{(x_1,x_2,x_3)\in \mathbb{R}^3: x_3> 0,\,x_1<0\},\\\\
 \mathbb{R}^3_l=\{(x_1,x_2,x_3)\in \mathbb{R}^3:
 x_1>0\},\end{array}\right.\end{equation}
for every $n\in \mathbb{N}$, problem (\ref{resequ}) will be
reformulated
 on a fixed domain through the following rescaling:
\begin{equation}\label{rescabis}{\cal T}_n:(x_1,x_2,x_3)\in \mathbb{R}^3\rightarrow {\cal T}_n(x_1,x_2,x_3)=\left\{\begin{array}{ll}(h_nx_1,h_nx_2,x_3), \hbox{ if }
(x_1,x_2,x_3)\in \mathbb{R}^3_{a,r},\\\\
(x_1,h_nx_2,h_nx_3), \hbox{ if }(x_1,x_2,x_3)\in
\mathbb{R}^3_l,\\\\(h_nx_1,h_nx_2,h_nx_3), \hbox{ if }
(x_1,x_2,x_3)\in \mathbb{R}^3_{b,r}
\end{array}\right.\end{equation}
(note that ${T_n}_{|\mathbb{R}^3_{b,r}}={{\cal
T}_n}_{|\mathbb{R}^3_{b,r}}$ , and $T_n(\Omega^a)={\cal
T}_n(\Omega^a)=\Omega^a_n$, $T_n(\Omega^b)={\cal
T}_n(\Omega^b)=\Omega_n^{b,l}$). Namely, setting
$$B^{a,r}_n=\left]-\frac{2 }{h_n}, 0\right[^2\times ]0,2[,\quad B^{l}_n=
]0, 2[\times \left]-\frac{2}{h_n},\frac{2}{h_n}\right[^2,\quad
B^{b,r}_n= \left]-\frac{2}{h_n},0\right[^3,\quad n\in
\mathbb{N},$$
    space ${\cal U}$ defined in
(\ref{calU}) is rescaled in the following
\begin{equation}\label{spazioriscalatoUw-wbis}\begin{array}{ll} {\cal V}_n= \big\{&v=(v^{a,r}, v^l,v^{b,r})
\in L^1_{loc} (\overline{\mathbb{R}^3_{a,r}})\times L^1_{loc}
(\overline{\mathbb{R}^3_l})\times L^1_{loc}
(\overline{\mathbb{R}^3_{b,r}})\,\,:\\\\&
(v^{a,r}_{|_{B^{a,r}_n}}, v^l_{|_{B^l_n}},v^{b,r}_{|_{B^{b,r}_n}})
\in L^2 (B^{a,r}_n)\times L^2 (B^l_n)\times L^2 (B^{b,r}_n),\\\\&
(Dv^{a,r},Dv^l,Dv^{b,r})\in
(L^2(\mathbb{R}^3_{a,r}))^3\times (L^2(\mathbb{R}^3_l))^3\times (L^2(\mathbb{R}^3_{b,r}))^3,\\\\
&\quad\displaystyle{\int_{B^{a,r}_n}v^{a,r}dx+\int_{B^l_n}v^ldx+h_n\int_{B^{b,r}_n}v^{b,r}dx=0,}\\\\&
v^l(0,x_2,x_3)=v^{a,r}(0,x_2,h_nx_3) ,\hbox{ for }(x_2,x_3) \hbox{
a.e. in }\mathbb{R}\times]0,+\infty[,\\\\
&v^l(0,x_2,x_3)= v^{b,r}(0,x_2,x_3),\hbox{ for }(x_2,x_3)\hbox{
a.e. in }\mathbb{R}\times]-\infty,0[,\\\\
&v^{a,r}(x_1,x_2,0)= v^{b,r}(x_1,x_2,0),\hbox{ for
}(x_1,x_2)\hbox{ a.e. in }]-\infty,0[\times\mathbb{R}
\big\}.\end{array}
\end{equation}  Then,  for every
$\underline{m}=(\underline{m}^a,\underline{m}^{b,l},\underline{m}^{b,r})\in
L^2(\Omega^a,\mathbb{R}^3)\times
L^2(\Omega^{b,l},\mathbb{R}^3)\times
L^2(\Omega^{b,r},\mathbb{R}^3)$, the following equation:
\begin{equation}\label{equariscalataw-wbis}\left\{
\begin{array}{l}v_{\underline{m},n}=(v_{\underline{m},n}^{a,r},v_{\underline{m},n}^l, v_{\underline{m},n}^{b,r})
\in {\cal V}_n,\\\\
 \displaystyle{\int_{\mathbb{R}^3_{a,r}}
\left(\frac{1}{h_n}D_{x_1}v_{\underline{m},n}^{a,r},\frac{1}{h_n}D_{x_2}v_{\underline{m},n}^{a,r},
D_{x_3}v_{\underline{m},n}^{a,r}\right)\left(\frac{1}{h_n}D_{x_1}v^{a,r},\frac{1}{h_n}D_{x_2}v^{a,r},
D_{x_3}v^{a,r}\right) dx+}\\\\
 \displaystyle{\int_{\mathbb{R}^3_l}
\left(D_{x_1}v_{\underline{m},n}^l,\frac{1}{h_n}D_{x_2}v_{\underline{m},n}^l,
\frac{1}{h_n}D_{x_3}v_{\underline{m},n}^l\right)\left(D_{x_1}v^l,\frac{1}{h_n}D_{x_2}v^l,
\frac{1}{h_n}D_{x_3}v^l\right)
dx+}\\\\
\displaystyle{\frac{1}{h_n}\int_{\mathbb{R}^3_{b,r}}
\left(D_{x_1}v_{\underline{m},n}^{b,r},D_{x_2}v_{\underline{m},n}^{b,r},
D_{x_3}v_{\underline{m},n}^{b,r}\right)\left(D_{x_1}v^{b,r},D_{x_2}v^{b,r},
D_{x_3}v^{b,r}\right)
dx=}\\\\\displaystyle{\int_{\Omega^a}\left(\frac{1}{h_n}D_{x_1}v^{a,r},\frac{1}{h_n}D_{x_2}v^{a,r},
D_{x_3}v^{a,r}\right){\underline{m}^a} dx+}\\\\
\displaystyle{\int_{\Omega^{b,l}}
\left(D_{x_1}v^l,\frac{1}{h_n}D_{x_2}v^l,
\frac{1}{h_n}D_{x_3}v^l\right){\underline{m}^{b,l}}
dx+}\\\\\displaystyle{\int_{\Omega^{b,r}}
\left(D_{x_1}v^{b,r},D_{x_2}v^{b,r},
D_{x_3}v^{b,r}\right){\underline{m}^{b,r}} dx,}\quad\forall
v=(v^{a,r},v^l,v^{b,r})\in {\cal
V}_n,\end{array}\right.\end{equation}
 which  rescales  equation (\ref{resequ}) by rescaling (\ref{rescabis}), admits a unique
solution.

For every $n\in\mathbb{ N}$, let $v_n=(v_n^{a,r},v_n^l,v_n^{b,r})$
be the unique solution of (\ref{equariscalataw-wbis})
corresponding to $\underline{m}_n$. Arguing as in the first part
of this proof, for a symmetric argument, one can easily prove that
\begin{equation}\label{star1}\begin{array}
{llll}D_{x_1}v_n^l\rightharpoonup0,\quad\dfrac{1}{h_n}D_{x_2}v_n^l\rightharpoonup
\xi^l_2,
 \quad  \dfrac{1}{h_n}D_{x_3}v_n^l\rightharpoonup\xi^l_3\hbox{ weakly in
 }L^2(\mathbb{R}^3_l),
\end{array}\end{equation}
as $n$ diverges, where
\begin{equation}\nonumber (\xi_2^l,\xi^l_3)(x_1,x_2,x_3)=
\left\{\begin{array}{ll} (0,0),
\hbox{ a.e. in } ]1,+\infty[\times \mathbb{R}^2, \\\\
\mu^b_2(x_1)Dp(x_2,x_3)+ \mu^b_3(x_1)Dq(x_2,x_3),\hbox{ a.e. in }
]0,1[\times \mathbb{R}^2,\end{array}\right.\end{equation} with $p$
(resp. $q$) the unique solution of (\ref{resequxx1}) (resp.
(\ref{resequxx2})). Furthermore, one has that
\begin{equation}\label{star2}\begin{array}{l}
\displaystyle{ \lim_n\int_{\Omega^b}
\left(D_{x_1}v_n^l,\frac{1}{h_n}D_{x_2}v_n^l,\frac{1}{h_n}D_{x_3}v_n^l,
\right)\underline{m}_n^b dx=}
\\\\\displaystyle{
\alpha(]-1,0[^2)\int_0^1\vert\mu^b_2\vert^2dx_1+\beta(]-1,0[^2)\int_0^1\vert\mu^b_3\vert^2dx_1+
\gamma(]-1,0[^2)\int_0^1\mu^b_2\mu^b_3dx_1,}
\end{array}\end{equation}
where  $\alpha(]-1,0[^2)$, $\beta(]-1,0[^2)$ and
$\gamma(]-1,0[^2)$ are defined by (\ref{abc}) with $S=]-1,0[^2$.

Now, to conclude it is enough to note that
\begin{equation}\nonumber {\cal
T}_n^{-1}(T_n(x))=x,\quad\forall x \in
\mathbb{R}^3\setminus\{(x_1,x_2,x_3)\in \mathbb{R}^3:x_1\geq0,\,
x_3\geq0\},\quad\forall n\in \mathbb{N},
\end{equation}
\begin{equation}\nonumber v_n\left({\cal
T}_n^{-1}(T_n(x))\right)=u_n(x),\quad \forall x \in
\mathbb{R}^3,\quad\forall n\in \mathbb{N}.
\end{equation}
Consequently, it results that
\begin{equation}\label{star3} v_n(x)=u_n(x),\quad \forall x \in
\mathbb{R}^3\setminus\{(x_1,x_2,x_3)\in \mathbb{R}^3:x_1\geq0,\,
x_3\geq0\},\quad\forall n\in \mathbb{N}.
\end{equation}
Then, combining (\ref{star1}) and (\ref{star2}) with
(\ref{star3}), one obtains the  fifth and the sixth  limit in
(\ref{conlimprinw-w}) and
\begin{equation}\label{star4}\begin{array}{l}
\displaystyle{ \lim_n\int_{\Omega^b}
\left(D_{x_1}u_n^l,\frac{1}{h_n}D_{x_2}u_n^l,\frac{1}{h_n}D_{x_3}u_n^l
\right)\underline{m}_n^b dx=}
\\\\\displaystyle{
\alpha(]-1,0[^2)\int_0^1\vert\mu^b_2\vert^2dx_1+\beta(]-1,0[^2)\int_0^1\vert\mu^b_3\vert^2dx_1+
\gamma(]-1,0[^2)\int_0^1\mu^b_2\mu^b_3dx_1.}
\end{array}\end{equation}
Finally, combining (\ref{limenw-wprimo}) with (\ref{star4}), also
limit (\ref{limenw-w}) holds true.
\end{proof}

\subsection{Proof of theorem \ref{ultimow-w}}
\begin{proof} By choosing $\underline{m}=\left((0,1,0),(0,1,0),(0,1,0)\right)$ as test
function in (\ref{problemariscalatow-w}),  taking into account
(\ref{forzew-w}) and that  $\vert \underline{m}_n\vert=1$ a.e. in
$\Omega^a\bigcup \Omega^{b,l}\bigcup \Omega^{b,r}$, using
proposition \ref{limvarprinw-w} and arguing as in the proof of
theorem \ref{ultimo}, it is easy to prove the existence of
$c\in]0,+\infty[$ such that
\begin{equation}\nonumber\left\{\begin{array}{l}\Vert D_{x_1}\underline{m}^a_n\Vert_{(L^2(\Omega^a))^3}\leq ch_n,\quad \Vert D_{x_2}\underline{m}^a_n\Vert_{(L^2(\Omega^a))^3}\leq ch_n, \quad
\Vert D_{x_3}\underline{m}^a_n\Vert_{(L^2(\Omega^a))^3}\leq c,
\\\\
\Vert
D_{x_1}\underline{m}^{b,l}_n\Vert_{(L^2(\Omega^{b,l}))^3}\leq
c,\quad \Vert
D_{x_2}\underline{m}^{b,l}_n\Vert_{(L^2(\Omega^{b,l}))^3}\leq
h_nc, \quad \Vert
D_{x_3}\underline{m}^{b,l}_n\Vert_{(L^2(\Omega^{b,l}))^3}\leq
ch_n, \\\\
\Vert D\underline{m}^{b,r}_n\Vert_{(L^2(\Omega^{b,l}))^9}\leq
c\sqrt{h_n} ,\end{array}\right.
\end{equation}
for every $n\in \mathbb{N}$. Then, taking into account again that
$\vert \underline{m}_n\vert=1$ a.e. in $\Omega^a\bigcup
\Omega^{b,l}\bigcup \Omega^{b,r}$, there exist an increasing
sequence of positive integer numbers $\{n_i\}_{i\in\mathbb N}$,
$\widehat{\mu}^a\in H^1(\Omega^a,S^2)$ independent of $(x_1,x_2)$,
$\widehat{\mu}^{b,l}\in H^1(\Omega^{b,l},S^2)$ independent of
$(x_2,x_3)$ and $c\in S^2$, $\zeta^a\in \left(L^2(\Omega^a,
\mathbb{R}^3)\right)^2$, $\zeta^{b,l}\in \left(L^2(\Omega^{b,l},
\mathbb{R}^3)\right)^2$, $\zeta^{b,r}\in \left(L^2(\Omega^{b,r},
\mathbb{R}^3)\right)^3$ such that
\begin{equation}\label{ifconvergenzamweakw-w}\left\{\begin{array}{l}
\underline{m}^{a}_{n_i}\rightharpoonup \widehat{\mu}^a\hbox{
weakly  in } H^1(\Omega^a,\mathbb{R}^3),\\\\
\underline{m}^{b,l}_{n_i}\rightharpoonup\widehat{\mu}^{b,l}\hbox{
weakly in } H^1(\Omega^{b,l},\mathbb{R}^3),\\\\
\underline{m}^{b,r}_{n_i}\rightharpoonup c\hbox{ weakly in }
H^1(\Omega^{b,r},\mathbb{R}^3),\end{array}\right.
\end{equation}
\begin{equation}\label{ifconvergenzamweakbozw-w}\left\{\begin{array}{l}\displaystyle{
\left(\frac{1}{h_{n_i}}D_{x_1}\underline{m}^a_{n_i},\frac{1}{h_{n_i}}
D_{x_2}\underline{m}^a_{n_i}\right) \rightharpoonup\zeta^a \hbox{
weakly  in }\left(L^2(\Omega^a,
\mathbb{R}^3)\right)^2,}\\\\\displaystyle{
\left(\frac{1}{h_{n_i}}D_{x_2}\underline{m}^{b,l}_{n_i},\frac{1}{h_{n_i}}
D_{x_3}\underline{m}^{b,l}_{n_i}\right) \rightharpoonup\zeta^{b,l}
\hbox{ weakly  in }\left(L^2(\Omega^{b,l},
\mathbb{R}^3)\right)^2,}\\\\
\displaystyle{ \frac{1}{\sqrt{h_{n_i}}}D\underline{m}^{b,r}
\rightharpoonup\zeta^{b,r} \hbox{ weakly  in
}\left(L^2(\Omega^{b,r}, \mathbb{R}^3)\right)^3,}
\end{array}\right.
\end{equation}
as $i$ diverges. Consequently, since one has that
\begin{equation}\nonumber\left\{\begin{array}{l}\underline{m}_n^a(x_1,x_2,0)=
\underline{m}_n^{b,r}(x_1,x_2,0),\hbox{ for }(x_1,x_2) \hbox{ a.e.
in }]-1,0[^2,\\\\ \underline{m}_n^{b,l}(0,x_2,x_3)=
\underline{m}_n^{b,r}(0,x_2,x_3),\hbox{ for }(x_2,x_3) \hbox{ a.e.
in }]-1,0[^2,\end{array}\right.
\end{equation} for every $n\in \mathbb{N}$, it follows that
$\widehat{\mu}^a(0)=c=\widehat{\mu}^{b,l}(0)$, that is
$\widehat{\mu}=(\widehat{\mu}^a,\widehat{\mu}^{b,l})\in {\cal M}$.
Moreover, by virtue of proposition \ref{limvarprinw-w}, limits in
(\ref{limconw-w}) hold true and it results that
\begin{equation}\label{conEmagniw-w}\begin{array}{l}
\displaystyle{ \lim_i\Bigg[\int_{\Omega^a}
\left(\frac{1}{h_{n_i}}D_{x_1}u_{n_i}^a,\frac{1}{h_{n_i}}D_{x_2}u_{n_i}^a,
D_{x_3}u_{n_i}^a\right)\underline{m}_{n_i}^a
dx+}\\\\\displaystyle{\int_{\Omega^{b,l}}
\left(D_{x_1}u_{n_i}^{b,l}, \frac{1}{h_{n_i}}D_{x_2}u_{n_i}^{b,l},
\frac{1}{h_{n_i}}D_{x_3}u_{n_i}^{b,l}\right)\underline{m}_{n_i}^{b,l}
dx+}\\\\\displaystyle{\int_{\Omega^{b,r}}
\left(D_{x_1}u_{n_i}^{b,r}, D_{x_2}u_{n_i}^{b,r},
D_{x_3}u_{n_i}^{b,r}\right)\underline{m}_{n_i}^{b,r} dx\Bigg]=}
\\\\\displaystyle{
\alpha(]-1,0[^2)\int_0^1\vert\mu^a_1\vert^2dx_3+\beta(]-1,0[^2)\int_0^1\vert\mu^a_2\vert^2dx_3+
\gamma(]-1,0[^2)\int_0^1\mu^a_1\mu^a_2dx_3+}\\\\\displaystyle{
\alpha(]-1,0[^2)\int_0^1\vert\mu^{b,l}_2\vert^2dx_1+\beta(]-1,0[^2)\int_0^1\vert\mu^{b,l}_3\vert^2dx_1+
\gamma(]-1,0[^2)\int_0^1\mu^{b,l}_2\mu^{b,l}_3dx_1,}
\end{array}\end{equation}
where  $\alpha(]-1,0[^2)$, $\beta(]-1,0[^2)$ and
$\gamma(]-1,0[^2)$ are defined by (\ref{abc}) with $S=]-1,0[^2$.

Now, the goal is to identify $\widehat{\mu}$, $\zeta^a$,
$\zeta^{b,l}$, $\zeta^{b,r}$, to obtain strong convergences in
(\ref{ifconvergenzamweakw-w}) and in
(\ref{ifconvergenzamweakbozw-w}), and to prove limit in
(\ref{convenergiew-w}). To this aim, for
$(\widehat{\mu}^a,\widehat{\mu}^{b,l})\in {\cal M}$, let us set
\begin{equation}\nonumber v=\left\{\begin{array}{l}\widehat{\mu}^a, \hbox{ in
}\Omega^a,\\\\
\widehat{\mu}^{b,l}, \hbox{ in }\Omega^{b,l},\\\\
\widehat{\mu}^a(0)=\widehat{\mu}^{b,l}(0), \hbox{ in
}\Omega^{b,r}.
\end{array}\right.\end{equation}
Obviously, $v\in {\cal M}_n$, for every $n\in N$. Then, by virtue
of l.s.c. arguments, (\ref{forzew-w}),
(\ref{ifconvergenzamweakw-w}), (\ref{ifconvergenzamweakbozw-w})
and (\ref{conEmagniw-w}) and proposition \ref{limvarprinw-w},   it
results that
\begin{equation}\label{lim2recvnzzw-w}\begin{array}{l}\displaystyle{\lambda\int_{\Omega^a}\vert
\zeta^a\vert^2dx +\lambda\int_{\Omega^{b,l}}\vert
\zeta^{b,l}\vert^2dx+\lambda\int_{\Omega^{b,r}}\vert
\zeta^{b,r}\vert^2dx+
E(\widehat{\mu}^a,\widehat{\mu}^b)\leq\liminf_i
E_{n_i}(\underline{m}_{n_i})\leq}\\\\\displaystyle{\limsup_i
E_{n_i}(\underline{m}_{n_i})\leq\lim_i
E_{n_i}(v)=E(\widehat{\mu}^a,\widehat{\mu}^b), \quad\forall
(\widehat{\mu}^a,\widehat{\mu}^{b,l})\in {\cal M}
.}\end{array}\end{equation}
 Consequently, $\zeta^a=0$,
$\zeta^{b,l}=0$, $\zeta^{b,r}=0$,
$(\widehat{\mu}^a,\widehat{\mu}^b)$ solves
(\ref{problemafinalew-w}) and limit (\ref{convenergiew-w}) holds
true. Finally, combining (\ref{convenergiew-w}) with
(\ref{forzew-w}), (\ref{ifconvergenzamweakw-w}),
(\ref{ifconvergenzamweakbozw-w}) and (\ref{conEmagniw-w}) one
obtains that limits in (\ref{ifconvergenzamweakw-w}) and in
(\ref{ifconvergenzamweakbozw-w}) are strong.
\end{proof}

\section*{Acknowledgments}
This work was supported by GNAMPA, by Universit\'e Paris Est and
by project "MICINN MTM 2009-12628".

\end{document}